\documentclass[fontsize=10pt]{scrartcl}
\usepackage[utf8]{inputenc}
\usepackage[T1]{fontenc}
\usepackage[american]{babel}
\usepackage[autostyle,english=american]{csquotes}
\usepackage[theoremfont]{newpxtext}
\usepackage{newpxmath}

\usepackage{amsthm}
\usepackage{amssymb}
\usepackage{amsmath}
\usepackage{enumitem}
\usepackage{mathtools}
\usepackage[headsepline]{scrlayer-scrpage}
\usepackage[pdfborder={0 0 0.5},pdfpagelabels,colorlinks,linkcolor=blue]{hyperref}
\newtheorem{theorem}{Theorem}
\newtheorem{lemma}[theorem]{Lemma}
\newtheorem{proposition}[theorem]{Proposition}
\newtheorem{corollary}[theorem]{Corollary}
\theoremstyle{definition}
\newtheorem{definition}[theorem]{Definition}
\newtheorem{remark}[theorem]{Remark}
\newtheorem{condition}[theorem]{Condition}
\numberwithin{theorem}{section}
\allowdisplaybreaks
\usepackage{cleveref}

\crefname{chapter}{Chapter}{Chapters}
\crefname{section}{Section}{Sections}
\crefname{subsection}{Section}{Sections}
\crefname{theorem}{Theorem}{Theorems}
\crefname{lemma}{Lemma}{Lemmas}
\crefname{proposition}{Proposition}{Propositions}
\crefname{corollary}{Corollary}{Corollaries}
\crefname{definition}{Definition}{Definitions}
\crefname{rem}{Remark}{Remarks}
\crefname{cond}{Condition}{Conditions}
\crefname{equation}{}{}
\crefname{enumi}{}{}

\newcommand{\wto}{\rightharpoonup}
\newcommand{\f}{f^\alpha}
\newcommand{\g}{g^\alpha}
\renewcommand{\a}{a^\alpha}
\newcommand{\e}{e_\alpha}

\renewcommand{\v}{\widehat{v}_\alpha}
\newcommand{\vo}{v_\alpha^0}
\newcommand{\gamm}{\gamma_\alpha}
\newcommand{\w}{w_\alpha}
\newcommand{\psia}{\psi^\alpha}
\newcommand{\taua}{\tau^\alpha}
\newcommand{\R}{\mathbb{R}}

\newcommand{\Y}{\mathcal Y_{\mathrm{pd}}^\alpha}
\newcommand{\Na}{\mathcal N_\alpha}
\newcommand{\suma}{\sum_{\alpha=1}^N}
\let\originalleft\left
\let\originalright\right
\renewcommand{\left}{\mathopen{}\mathclose\bgroup\originalleft}
\renewcommand{\right}{\aftergroup\egroup\originalright}

\DeclareMathOperator{\supp}{supp}
\let\div\relax
\DeclareMathOperator{\div}{div}
\DeclareMathOperator{\curl}{curl}
\newcommand{\inte}{\mathrm{int}}

\DeclareMathOperator{\dist}{dist}
\DeclareMathOperator{\signum}{sign}
\newcommand{\sign}[1]{\signum\left(#1\right)}
\newcommand{\kin}{{\alpha\mathrm{kin}}}
\newcommand{\loc}{\mathrm{loc}}
\automark{section}
\automark*{subsection}
\ihead{\leftmark}
\ohead{\thepage}
\begin{document}

\title{Optimal Control of a Hot Plasma}
\author{Jörg Weber\\ \textit{University of Bayreuth, 95440 Bayreuth, Bavaria, Germany}\\ \texttt{Joerg.Weber@uni-bayreuth.de}}
\date{}
\maketitle
\begin{abstract}
	The time evolution of a collisionless plasma is modeled by the relativistic Vlasov-Maxwell system which couples the Vlasov equation (the transport equation) with the Maxwell equations of electrodynamics. We consider the case that the plasma is located in a bounded container $\Omega\subset\R^3$, for example a fusion reactor. Furthermore, there are external currents, typically in the exterior of the container, that may serve as a control of the plasma if adjusted suitably. We model objects, that are placed in space, via given matrix-valued functions $\varepsilon$ (the permittivity) and $\mu$ (the permeability). A typical aim in fusion plasma physics is to keep the amount of particles hitting $\partial\Omega$ as small as possible (since they damage the reactor wall), while the control costs should not be too exhaustive (to ensure efficiency). This leads to a minimizing problem with a PDE constraint. This problem is analyzed in detail. In particular, we prove existence of minimizers and establish an approach to derive first order optimality conditions.
	
	\vspace*{3mm}
	
	
	\vspace*{3mm}
	
\end{abstract}

\numberwithin{equation}{section}
\section{Introduction}
If a plasma is sufficiently rarefied or hot, collisions among the plasma particles can be neglected and the time evolution of this plasma can be modeled by the relativistic Vlasov-Maxwell system. In our set-up, this system reads as follows:
\begin{subequations}\renewcommand{\theequation}{VM.\arabic{equation}}\phantomsection\makeatletter\def\@currentlabel{VM}\label{eq:WholeSystem}\makeatother
	\begin{align}
	\partial_t\f+\v\cdot\partial_x\f+\e\left(E+\v\times H\right)\cdot\partial_v\f&=0&\ \mathrm{on}\ \left[0,T\right]\times\Omega\times\R^3,\label{eq:WholeVl}\\
	\f_-&=\a\left(Kf_+\right)&\ \mathrm{on}\ \gamma_{T}^-,\label{eq:WholeBoun}\\
	\f\left(0\right)&=\mathring\f&\ \mathrm{on}\ \Omega\times\R^3,\label{eq:WholeInitf}\\
	\varepsilon\partial_tE-\curl_xH&=-4\pi j&\ \mathrm{on}\ \left[0,T\right]\times\R^3,\label{eq:WholeMax1}\\
	\mu\partial_tH+\curl_xE&=0&\ \mathrm{on}\ \left[0,T\right]\times\R^3,\label{eq:WholeMax2}\\
	\left(E,H\right)\left(0\right)&=\left(\mathring E,\mathring H\right)&\ \mathrm{on}\ \R^3.\label{eq:WholeInitEH}
	\end{align}
\end{subequations}
This set of equations describes the time evolution of a collisionless plasma which consists of $N$ particle species and is located in some bounded domain $\Omega\subset\R^3$. Equations \cref{eq:WholeVl,eq:WholeBoun,eq:WholeInitf} are to hold for each $\alpha=1,\dots,N$, where \cref{eq:WholeVl} is the Vlasov equation for the density $\f=\f\left(t,x,v\right)$ of the $\alpha$-th particle species. These densities depend on time $t\in\left[0,T\right]$, where $T>0$ is some given final time, on position $x\in\Omega$ and momentum $v\in\R^3$, from which the relativistic velocity is computed via
\begin{align*}
\v=\frac{v}{\sqrt{m_\alpha^2+\left|v\right|^2}}.
\end{align*}
The quantities $m_\alpha$ and $\e$ are the rest mass and charge of a particle of the $\alpha$-th species.

Equation \cref{eq:WholeInitf} is the initial condition for $\f$ and \cref{eq:WholeBoun} describes the boundary condition on $\partial\Omega$. Here, $\f_\pm$ can be understood as the restrictions of $\f$ to
\begin{align*}
\gamma_{T}^\pm&\coloneqq\left\{\left(t,x,v\right)\in\left[0,T\right]\times\partial\Omega\times\R^3\mid v\cdot n\left(x\right)\gtrless 0\right\},
\end{align*}
$K$ typically describes reflection on $\partial\Omega$ via
\begin{align*}
\left(Kh\right)\left(t,x,v\right)=h\left(t,x,v-2\left(v\cdot n\left(x\right)\right)\right),
\end{align*}
and $\a=\a\left(t,x,v\right)$ is a factor; pure reflection corresponds to $\a=1$ and absorption to $0\leq\a<1$. Above, $n\left(x\right)$ denotes the outer unit normal of $\partial\Omega$ at $x\in\partial\Omega$.

Equations \cref{eq:WholeMax1,eq:WholeMax2} are the time-evolutionary Maxwell equations for the electromagnetic fields $E=E\left(t,x\right)$, $H=H\left(t,x\right)$ with initial condition \cref{eq:WholeInitEH}. The source term
\begin{align*}
j\coloneqq j^\inte+u\coloneqq\suma\e\int_{\R^3}\v\f\,dv+u
\end{align*}
is the sum of the internal and external currents $j^\inte$ and $u$. Throughout this paper, we assume that $u$ is supported in some bounded set $\Gamma\subset\R^3$ -- for example, $\Gamma$ is the set where the electric coils containing the external currents are placed in space. Note that \cref{eq:WholeMax1,eq:WholeMax2} are imposed on whole space -- in contrast to \cref{eq:WholeVl} -- to allow interaction of the exterior and interior of $\Omega$. Hence, we do not prescribe perfect conductor boundary conditions on $\partial\Omega$ for the electromagnetic fields, but rather model objects that are placed somewhere in space, for example the reactor wall, electric coils, or (almost perfect) superconductors, by $\varepsilon$ (the permittivity) and $\mu$ (the permeability), which are given functions of the space coordinate, take values in the set of symmetric, positive definite matrices of dimension three, and do not depend on time. Thus, we make the following simplifying assumption on the constitutive equations: $D=\varepsilon E$, $B=\mu H$. With this assumption we can model linear, possibly anisotropic materials that stay fixed in time and whose material parameters are independent of the applied electromagnetic fields.

To ensure that the speed of light is constant in $\Omega$ and hence ensure that $\v$ is independent of $x$, we have to assume that $\varepsilon\mu$ is constant in $\Omega$. Throughout this work we use modified Gaussian units such that $\varepsilon=\mu=1$ on $\Omega$ -- thus, the speed of light is $1$ in $\Omega$ -- and all rest masses $m_\alpha$ of a particle of the respective species are at least $1$. Clearly, $\left|\v\right|<1$, that is, the velocity of a particle is bounded by the speed of light (in $\Omega$).

For a more detailed introduction we refer to \cite{Web19}. There, results on existence of weak solutions to \eqref{eq:WholeSystem}, which will be stated later, as well as a proof of the redundancy of the divergence part of Maxwell's equations, which we therefore neglect in this work, in a weak solution concept can be found.

In this paper, we analyze a minimizing problem where a certain objective function shall be driven to a minimum over a certain set of functions satisfying \eqref{eq:WholeSystem} in a weak sense. More precisely, the objective function is
\begin{align*}
\frac{1}{q}\suma \w\left\|\f_+\right\|_{L^q\left(\gamma_{T}^+,d\gamm\right)}^q+\frac{1}{r}\left\|u\right\|_{\mathcal U}^r.
\end{align*}
Here, $2<q<\infty$, $\w>0$, $\mathcal U=W^{1,r}\left(\left[0,T\right]\times\Gamma;\R^3\right)$ with $\frac{4}{3}<r<\infty$, and
\begin{align*}
d\gamm=\left|\v\cdot n\left(x\right)\right|\,dvdS_xdt
\end{align*}
is a surface measure on $\left[0,\infty\right[\times\partial\Omega\times\R^3$, which arises also canonically in the weak formulation of \eqref{eq:WholeSystem}. Thus, the objective function penalizes hits of the particles on $\partial\Omega$ (such hits usually damage the reactor wall) and exhaustive control costs (to ensure efficiency of a reactor). In addition to \eqref{eq:WholeSystem}, it is necessary to impose two inequality constraints, namely \cref{eq:constrfinfty,eq:constrener}.

The paper is organized as follows: In \cref{sec:preliminaries}, we state the weak formulation of \eqref{eq:WholeSystem} and the results of \cite{Web19} which will be needed later. In \cref{sec:problem}, we discuss the minimizing problem in detail. After that, we firstly prove existence of a minimizer in \cref{sec:ExistenceMinimizer}, see \cref{thm:ExistenceMinimizer}. Secondly, we establish an approach to derive first order optimality conditions for a minimizer in \cref{sec:weakrevisited,sec:firstorder}. To this end, the one main idea is to write the weak form of \eqref{eq:WholeSystem} equivalently as an identity
\begin{align*}
\mathcal G\left(\left(\f,\f_+\right)_\alpha,E,H,u\right)=0\mbox{ in }\Lambda^*,
\end{align*}
where $\mathcal G$ is differentiable and $\Lambda$ is a uniformly convex, reflexive test function space; see \cref{sec:weakrevisited}. The other main idea is to introduce an approximate minimizing problem with a penalizing parameter $s>0$ which is driven to infinity later, see \cref{sec:firstorder}. In particular, we add the differentiable term \begin{align*}
\frac{s}{2}\left\|\mathcal G\left(\left(\f,\f_+\right)_\alpha,E,H,u\right)\right\|_{\Lambda^*}^2
\end{align*}
to the original objective function and abolish the constraint that \eqref{eq:WholeSystem} be solved. For this approximate problem, we prove existence of a minimizer and establish first order optimality conditions, see \cref{thm:ExistenceMinimizerPs,thm:OCs}. After that, we let $s\to\infty$ and prove that, along a suitable sequence, a minimizer of the original problem is obtained in the limit, and the convergence of the controls $u$ is even strong, see \cref{thm:passlimit}. Lastly, we briefly discuss in \cref{sec:finalremarks} how these results can also be verified in case of similar set-ups or different objective functions.

We should point out that the main problem we have to deal with is that existence of global-in-time solutions to \eqref{eq:WholeSystem} is only known in a weak solution concept. It is an open problem whether or not such weak solutions are unique for given $u$. Thus, standard approaches to derive first order optimality conditions via introducing a (preferably differentiable) control-to-state operator can not be applied.

Vlasov-Maxwell systems have been studied extensively. Proofs for global existence of classical -- i.e., $C^1$ -- solutions, in lower dimensional settings can be found in works of Glassey and Schaeffer \cite{GS90,GS97,GS98a,GS98b}. Global well-posedness of the Cauchy problem in three dimensions, however, is a famous open problem. Weak solutions were constructed by Di Perna and Lions \cite{DL89} and -- in case of the presence of boundary conditions, including perfect conductor boundary conditions for the electromagnetic fields -- by Guo \cite{Guo93}. For existence of weak solutions, a compactness result of \cite{DL89} is fundamental, since it handles the nonlinearity in the Vlasov equation, and will also be applied in this paper. For a more detailed overview we refer to Rein \cite{Rei04}.

Controllability of the relativistic Vlasov-Maxwell system in two dimensions has been studied by Glass and Han-Kwan \cite{GH15}. Knopf \cite{Kno18} and later Knopf and the author \cite{KW18} analyzed optimal control problems for the Vlasov-Poisson system, where Maxwell's equations are replaced by the electrostatic Poisson equation. Here, an external magnetic field was considered as a control. Studying control problems with the Vlasov-Poisson system as the governing PDE system enjoys the advantage of having existence and uniqueness of global-in-time classical solutions on hand, due to the results of Pfaffelmoser \cite{Pfa92} and Schaeffer \cite{Sch91}. Also, an optimal control problem for the two-dimensional Vlasov-Maxwell system was considered in \cite{Web18}.

\section{Weak formulation and existence results}\label{sec:preliminaries}
In the following, we state some results of \cite{Web19} which will be needed later.

The space of test functions for \cref{eq:WholeVl,eq:WholeBoun,eq:WholeInitf} is
\begin{align*}
\Psi_{T}&\coloneqq\left\{\psi\in C^\infty\left(\left[0,T\right]\times\overline\Omega\times\R^3\right)\mid\right.\supp\psi\subset\left[0,T\right[\times\overline\Omega\times\R^3\ \mathrm{compact},\phantom{\,\,\}}\nonumber\\
&\omit\hfill$\displaystyle\left.\dist\left(\supp\psi,\gamma_{T}^0\right)>0,\vphantom{\left(\left[0,T\right]\times\overline\Omega\times\R^3\right)}\dist\left(\supp\psi,\left\{0\right\}\times\partial\Omega\times\R^3\right)>0\right\}.$
\end{align*}
On the other hand,
\begin{align*}
\Theta_{T}\coloneqq\left\{\vartheta\in C^\infty\left(\left[0,T\right]\times \R^3;\R^3\right)\mid\supp\vartheta\subset\left[0,T\right[\times\R^3\ \mathrm{compact}\right\}
\end{align*}
is the space of test functions for \cref{eq:WholeMax1,eq:WholeMax2,eq:WholeInitEH}.

\begin{definition}\label{def:WeakSolWholeSys}
	We call a tuple $\left(\left(\f,\f_+\right)_\alpha,E,H,j\right)$ a weak solution of \eqref{eq:WholeSystem} on the time interval $\left[0,T\right]$ if (for all $\alpha$)
	\begin{enumerate}[label=(\roman*),leftmargin=*]
		\item $\f\in L_\loc^1\left(\left[0,T\right]\times\overline\Omega\times\R^3\right)$, $\f_+\in L_\loc^1\left(\gamma_{T}^+,d\gamm\right)$, $E,H,j\in L_\loc^1\left(\left[0,T\right]\times\R^3;\R^3\right)$;
		\item\label{def:WeakSolWholeSysii} for all $\psi\in\Psi_{T}$ it holds that
		\begin{align}\label{eq:Vlasovweak}
		0&=-\int_0^{T}\int_\Omega\int_{\R^3}\left(\partial_t\psi+\v\cdot\partial_x\psi+\e\left(E+\v\times H\right)\cdot\partial_v\psi\right)\f\,dvdxdt\nonumber\\
		&\phantom{=\;}+\int_{\gamma_{T}^+}\f_+\psi\,d\gamm-\int_{\gamma_{T}^-}\a\left(K\f_+\right)\psi\,d\gamm-\int_\Omega\int_{\R^3}\mathring\f\psi\left(0\right)\,dvdx
		\end{align}
		(in particular, especially the integral of $\left(E+\v\times H\right)\f\cdot\partial_v\psi$ is supposed to exist);
		\item\label{def:WeakSolWholeSysiii} for all $\vartheta\in\Theta_{T}$ it holds that
		\begin{subequations}\label{eq:Maxwellweak}
			\begin{align}
			0&=\int_0^{T}\int_{\R^3}\left(\varepsilon E\cdot\partial_t\vartheta-H\cdot\curl_x\vartheta-4\pi j\cdot\vartheta\right)\,dxdt+\int_{\R^3}\varepsilon\mathring{E}\cdot\vartheta\left(0\right)\,dx,\label{eq:Maxwellweak1}\\
			0&=\int_0^{T}\int_{\R^3}\left(\mu H\cdot\partial_t\vartheta+E\cdot\curl_x\vartheta\right)\,dxdt+\int_{\R^3}\mu\mathring{H}\cdot\vartheta\left(0\right)\,dx.
			\end{align}
		\end{subequations}
	\end{enumerate}
\end{definition}
Throughout this paper, we assume the following:
\begin{condition}\label{cond:data} Assume
	\begin{itemize}[leftmargin=*]
		\item $0\leq\mathring\f\in\left(L_\kin^1\cap L^\infty\right)\left(\Omega\times\R^3\right)$, $\mathring\f\not\equiv 0$ for all $\alpha=1,\dots,N$;
		\item $0\leq\a\in L^\infty\left(\gamma_{T}^-\right)$, $\a_0\coloneqq\|\a\|_{L^\infty\left(\gamma_{T}^-\right)}<1$ for all $\alpha=1,\dots,N$;
		\item $\mathring E,\mathring H\in L^2\left(\R^3;\R^3\right)$;
		\item $\varepsilon,\mu\in L^\infty\left(\R^3;\R^{3\times 3}\right)$ such that there are $\sigma,\sigma'>0$ satisfying $\sigma\leq\varepsilon,\mu\leq\sigma'$, and $\varepsilon=\mu=1$ on $\Omega$;
		\item $u\in L^1\left(\left[0,T\right];L^2\left(\Gamma;\R^3\right)\right)$;
		\item $\partial\Omega$ is of class $C^{1,\kappa}$ for some $\kappa>0$.
	\end{itemize}
\end{condition}
Here and in the following, we define for $\alpha=1,\dots,N$
\begin{align*}
L_\kin^1\left(A,da\right)\coloneqq\left\{u\in L^1\left(A,da\right)\mid\int_A\vo\left|u\right|\,da<\infty\right\},
\end{align*}
equipped with the corresponding weighted norm, where $A\subset\R^3\times\R^3$ or $A\subset\R\times\R^3\times\R^3$ is some Borel set equipped with a measure $a$ and the weight $\vo$ is given by
\begin{align*}
\vo\coloneqq\sqrt{m_\alpha^2+\left|v\right|^2}.
\end{align*}
By $m_\alpha\geq 1$ we have $\vo\geq 1$. If $a$ is the Lebesgue measure, we write $L_\kin^1\left(A\right)$. Furthermore, $\sigma\leq\varepsilon\leq\sigma'$ (and likewise for $\mu$) is shorthand for: For almost any $x\in\R^3$ there holds $\sigma\left|E\right|^2\leq\varepsilon E\cdot E\leq\sigma'\left|E\right|^2$ for any $E\in\R^3$.

Note that it is necessary for the following result and for later considerations to assume partially absorbing boundary conditions, i.e., $\a_0<1$.
\begin{proposition}\label{prop:Existence}
	Let \cref{cond:data} hold. Then (for all $\alpha$) there exist functions
	\begin{itemize}[leftmargin=*]
		\item $\f\in L^\infty\left(\left[0,T\right];\left(L_\kin^1\cap L^\infty\right)\left(\Omega\times\R^3\right)\right),\f_+\in\left(L_\kin^1\cap L^\infty\right)\left(\gamma_{T}^+,d\gamm\right)$, all nonnegative,
		\item $\left(E,H\right)\in L^\infty\left(\left[0,T\right];L^2\left(\R^3;\R^6\right)\right)$
	\end{itemize}
	such that $\left(\left(\f,\f_+\right)_\alpha,E,H,j\right)$ is a weak solution of \eqref{eq:WholeSystem} on the time interval $\left[0,T\right]$ in the sense of \cref{def:WeakSolWholeSys}, where
	\begin{align*}
	j&=j^\inte+u=\suma\e\int_{\R^3}\v\f\,dv+u,\quad j^\inte\in L^\infty\left(\left[0,T\right];\left(L^1\cap L^{\frac{4}{3}}\right)\left(\Omega;\R^3\right)\right).
	\end{align*}
	Furthermore, we have the following estimates for $1\leq p\leq\infty$ and $t\in\left[0,T\right]$:\\
	Estimates on $\f,\f_+$:
	\begin{align}
	\left\|\f\right\|_{L^\infty\left(\left[0,t\right];L^p\left(\Omega\times\R^3\right)\right)}&\leq\left\|\mathring \f\right\|_{L^p\left(\Omega\times\R^3\right)},\label{eq:estf1}\\
	\left\|\f_+\right\|_{L^p\left(\gamma_t^+,d\gamm\right)}&\leq\left(1-\a_0\right)^{-\frac{1}{p}}\left\|\mathring \f\right\|_{L^p\left(\Omega\times\R^3\right)},\label{eq:estf+}
	\end{align}
	Energy-like estimate:
	\begin{align}\label{eq:estEner}
	&\left(\suma\left(1-\a_0\right)\int_{\gamma_t^+}\vo \f_+\,d\gamm\vphantom{\left\|\suma\int_\Omega\int_{\R^3}\vo\f\left(\cdot\right)\,dvdx+\frac{\sigma}{8\pi}\left\|\left(E,H\right)\left(\cdot\right)\right\|_{L^2\left(\R^3;\R^6\right)}^2\right\|_{L^\infty\left(\left[0,T\right]\right)}}+\left\|\suma\int_\Omega\int_{\R^3}\vo\f\left(\cdot\right)\,dvdx+\frac{\sigma}{8\pi}\left\|\left(E,H\right)\left(\cdot\right)\right\|_{L^2\left(\R^3;\R^6\right)}^2\right\|_{L^\infty\left(\left[0,t\right]\right)}\right)^{\frac{1}{2}}\nonumber\\
	&\leq\left(\suma\int_\Omega\int_{\R^3}\vo\mathring\f\,dvdx+\frac{\sigma'}{8\pi}\left\|\left(\mathring E,\mathring H\right)\right\|_{L^2\left(\R^3;\R^6\right)}^2\right)^{\frac{1}{2}}+\sqrt{2\pi}\sigma^{-\frac{1}{2}}\left\|u\right\|_{L^1\left(\left[0,t\right];L^2\left(\Gamma;\R^3\right)\right)}.
	\end{align}
\end{proposition}
Note that above we used the somewhat sloppy notation
\begin{align*}
L^\infty\left(\left[0,T\right];L^\infty\left(\Omega\times\R^3\right)\right)\coloneqq L^\infty\left(\left[0,T\right]\times\Omega\times\R^3\right).
\end{align*}

\section{The minimizing problem}\label{sec:problem}
In a fusion reactor, the main goal is to keep the particles away from the boundary of their container $\Omega$ since particles hitting the boundary damage the material there due to the usually very hot temperature of the plasma. Therefore, it is reasonable to penalize these hits, which, for example, can be achieved by taking some $L^q$-norms of the $\f_+$ as a part of the objective function that shall be minimized in an optimal control problem. We want to control the hits on the boundary by suitably adjusting the external current $u$.

Apart from driving the amount of hits on the boundary to a minimum, one does not want too exhaustive control costs so that the fusion reactor may have a good efficiency. Thus, we have to add some norm of $u$ to the objective function. Thereby, we also gain a mathematical advantage since then the objective function is coercive in $u$, which means that along a minimizing sequence this $u$-norm is bounded so that we can hope for being able to extract a weakly convergent subsequence whose weak limit is a candidate for an optimal control.

Conversely, as there are no terms including $\f$, $E$, and $H$ in the objective function, we do not have coercivity in these state variables only because of the objective function. But there is still the PDE system \eqref{eq:WholeSystem} as a constraint. Having a look at \cref{eq:estf1,eq:estf+,eq:estEner} we see that these estimates yield uniform boundedness of $\f$, $E$, $H$ in various norms along a minimizing sequence. Unfortunately, we can only verify these estimates for weak solutions that are constructed as in \cite{Web19}. For general weak solutions of \eqref{eq:WholeSystem} in the sense of \cref{def:WeakSolWholeSys} these estimates may be violated as we do not know a way to prove these generally. Since in the classical context these estimates are easily heuristically established by exploiting an energy balance and the measure preserving nature of the characteristic flow of the Vlasov equation, it is reasonable to restrict ourselves to weak solutions that satisfy at least part of, maybe slightly weaker versions of \cref{eq:estf1,eq:estf+,eq:estEner}.

To put our hands on the fields, only \cref{eq:estEner} is helpful. Considering this estimate along a minimizing, weakly converging sequence and trying to pass to the limit in this estimate, we see that the right-hand side, including some norm of $u$, has to be weakly continuous. But if we endow the control space with the norm that appears in \cref{eq:estEner}, i.e., the $L^1\left(\left[0,T\right];L^2\left(\Gamma;\R^3\right)\right)$-norm, this weak continuity will not hold. Consequently, we consider a control space that is compactly embedded in $L^1\left(\left[0,T\right];L^2\left(\Gamma;\R^3\right)\right)$, so that the right-hand side of \cref{eq:estEner} converges even if the controls only converge weakly in this new stronger control space. This will be made clear in the proof of \cref{thm:ExistenceMinimizer}.

Altogether, we arrive at the following minimizing problem:
\begin{align}\tag{P}\label{prob:P}
\left.
\begin{aligned}
\min_{y\in\mathcal Y,u\in\mathcal U}\quad&\mathcal J\left(y,u\right)=\frac{1}{q}\suma \w\left\|\f_+\right\|_{L^q\left(\gamma_{T}^+,d\gamm\right)}^q+\frac{1}{r}\left\|u\right\|_{\mathcal U}^r,\\
\mathrm{s.t.}\quad&\left(\left(\f,\f_+\right)_\alpha,E,H,j^\inte+u\right)\ \mathrm{solves}\ \eqref{eq:WholeSystem},\\
&\cref{eq:constrfinfty}\ \mathrm{and}\ \cref{eq:constrener}\ \mathrm{hold}
\end{aligned}
\right\}
\end{align}
where the additional constraints are
\begin{align}
&0\leq\f\leq\left\|\mathring\f\right\|_{L^\infty\left(\Omega\times\R^3\right)}\ \mathrm{a.e.},\alpha=1,\dots,N,\label{eq:constrfinfty}\\
&\suma\int_0^{T}\int_\Omega\int_{\R^3}\vo\f\,dvdxdt+\frac{\sigma}{8\pi}\int_0^{T}\int_{\R^3}\left(\left|E\right|^2+\left|H\right|^2\right)\,dxdt\nonumber\\
&\leq 2T\suma\int_\Omega\int_{\R^3}\vo\mathring\f\,dvdx+\frac{T\sigma'}{4\pi}\left\|\left(\mathring E,\mathring H\right)\right\|_{L^2\left(\R^3;\R^6\right)}^2+2\pi T^2\sigma^{-1}\left\|u\right\|_{L^2\left(\left[0,T\right]\times\Gamma;\R^3\right)}^2\label{eq:constrener}\\
&\eqqcolon\mathcal I\left(u\right).\nonumber
\end{align}
\begin{remark}
	We explain the formulation of the minimizing problem in detail:
	\begin{itemize}[leftmargin=*]
		\item We consider the optimal control problem on a finite time interval, i.e., $T<\infty$.
		\item For ease of notation, we have abbreviated
		\begin{align*}
		y=\left(\left(\f,\f_+\right)_\alpha,E,H\right),\quad
		\mathcal Y=\left(\bigtimes_{\alpha=1}^N\mathcal Y_{\mathrm{pd}}^\alpha\times L^q\left(\gamma_{T}^+,d\gamm\right)\right)\times L^2\left(\left[0,T\right]\times\R^3;\R^3\right)^2,
		\end{align*}
		where $1<q<\infty$ is fixed and
		\begin{align*}
		\Y\coloneqq\left\{f\in\left(L_\kin^1\cap L^\infty\right)\left(\left[0,T\right]\times\Omega\times\R^3\right)\mid\partial_tf+\v\cdot\partial_xf\in L^2\left(\left[0,T\right]\times\Omega;H^{-1}\left(\R^3\right)\right)\right\}.
		\end{align*}
		In the following, we denote
		\begin{align*}
		\Na\left(f\right)\coloneqq\left\|\partial_tf+\v\cdot\partial_xf\right\|_{L^2\left(\left[0,T\right]\times\Omega;H^{-1}\left(\R^3\right)\right)}.
		\end{align*}
		The restriction in the definition of $\Y$ will not be important until \cref{sec:firstorder} and is motivated by \cref{lma:etafinL2H-1}, which is stated below.
		\item The $\w>0$ are weights. For example, if we have two sorts of particles, ions and electrons say, the weight corresponding to the ions should be larger than the one corresponding to the electrons since the heavy ions will cause more damage on the boundary of a fusion reactor if they hit it. Moreover, the weights also serve as an indicator of which of our two aims should rather be achieved, that is to say no hits on the boundary and low control costs. More precisely, the $\w$ should be large if one rather wants no hits on the boundary, and should be small if one rather wants small control costs.
		\item The control space is
		\begin{align*}
		\mathcal U=W^{1,r}\left(\left[0,T\right]\times\Gamma;\R^3\right)
		\end{align*}
		where $\frac{4}{3}<r<\infty$ is fixed. By Sobolev's embedding theorem, $\mathcal U$ is compactly embedded in $L^2\left(\left[0,T\right]\times\Gamma;\R^3\right)$. For this, the boundary of $\Gamma$ has to satisfy some regularity condition, for example the cone condition. From now on, we shall always assume that $\partial\Gamma$ is not \enquote{too bad}, that is to say we have the compact embedding stated above. We endow $\mathcal U$ with the norm
		\begin{align*}
		\left\|u\right\|_{\mathcal U}\coloneqq\left(\sum_{j=1}^3\int_0^{T}\int_\Gamma\left(\left|u_j\right|^r+\kappa_1\left|\partial_tu_j\right|^r+\kappa_2\sum_{i=1}^3
		\left|\partial_{x_i}u_j\right|^r\right)\,dxdt\right)^{\frac{1}{r}},
		\end{align*}
		which is equivalent to the standard $W^{1,r}\left(\left[0,T\right]\times\Gamma;\R^3\right)$-norm. Here, $\kappa_1,\kappa_2>0$ are parameters chosen according to how much one wants to penalize $u$ itself compared to its $t$-and $x$-derivatives.
		\item As usual,
		\begin{align*}
		j^\inte=\suma\e\int_{\R^3}\v\f\,dv.
		\end{align*}
		\item The constraint that \eqref{eq:WholeSystem} be solved is to be understood in the sense of \cref{def:WeakSolWholeSys}.
		\item The pointwise constraint of $\f$ is on the one hand natural since any classical solution of \cref{eq:WholeVl} with nonnegative initial datum satisfies this constraint -- and also the weak solutions of \cref{prop:Existence} do -- and on the other hand necessary for a limit process when proving existence of a minimizer, see \cref{sec:ExistenceMinimizer}.
		\item The same applies \textit{mutatis mutandis} for the energy constraint. Note that this inequality directly follows from the stronger inequality \cref{eq:estEner} after an integration in time and Hölder's inequality:
		\begin{align*}
		&\suma\int_0^{T}\int_\Omega\int_{\R^3}\vo\f\,dvdxdt+\frac{\sigma}{8\pi}\int_0^{T}\int_{\R^3}\left(\left|E\right|^2+\left|H\right|^2\right)\,dxdt\\
		&\leq\int_0^{T}\left\|\suma\int_\Omega\int_{\R^3}\vo\f\left(\cdot\right)\,dvdx+\frac{\sigma}{8\pi}\left\|\left(E,H\right)\left(\cdot\right)\right\|_{L^2\left(\R^3;\R^6\right)}^2\right\|_{L^\infty\left(\left[0,s\right]\right)}\,ds\\
		&\leq 2T\left(\suma\int_\Omega\int_{\R^3}\vo\mathring\f\,dvdx+\frac{\sigma'}{8\pi}\left\|\left(\mathring E,\mathring H\right)\right\|_{L^2\left(\R^3;\R^6\right)}^2\right)\\
		&\phantom{\leq\;}+4\pi\sigma^{-1}\int_0^{T}\left(\int_0^s\left\|u\left(t\right)\right\|_{L^2\left(\Gamma;\R^3\right)}\,dt\right)^2ds
		\end{align*}
		and
		\begin{align*}
		\int_0^{T}\left(\int_0^s\left\|u\left(t\right)\right\|_{L^2\left(\Gamma;\R^3\right)}\,dt\right)^2ds\leq\int_0^{T}s\int_0^s\left\|u\left(t\right)\right\|_{L^2\left(\Gamma;\R^3\right)}^2\,dtds\leq\frac{T^2}{2}\left\|u\right\|_{L^2\left(\left[0,T\right]\times\Gamma;\R^3\right)}^2.
		\end{align*}
		The main reason why we impose the weaker inequality \cref{eq:constrener} as a constraint is that no longer $L^\infty$-terms or square roots appear, which would cause some trouble with respect to differentiability.
	\end{itemize}
\end{remark}

\section{Existence of minimizers}\label{sec:ExistenceMinimizer}
The usual strategy to obtain a minimizer of an optimization problem is to consider a minimizing sequence. By structure of the objective function or the constraints, this sequence is bounded in some norm so that we can extract a weakly converging subsequence (of course, we have to work in a reflexive space for this). To pass to the limit in a nonlinear optimization problem, some compactness is needed. As for passing to the limit in a nonlinear PDE (system), usually the same tools have to be exploited that were established to be able to pass to the limit in an iteration scheme to prove existence of solutions to the PDE (system).

This general strategy also applies to our case. The crucial compactness result is the following momentum averaging lemma by Di Perna and Lions \cite{DL89}; see also \cite{Rei04} for a shortened proof:
\begin{lemma}\label{lma:MomentumAveraging}
	Let $r>0$ and $\zeta\in C_c^\infty\left(B_r\right)$. There exists a constant $C>0$ such that for any functions $h,g_0,g_1\in L^2\left(\R\times\R^3\times B_r\right)$ which satisfy the inhomogeneous transport equation
	\begin{align*}
	\partial_th+\v\cdot\partial_xh=g_0+\div_vg_1
	\end{align*}
	in the sense of distributions we have
	\begin{align*}
	\int_{B_r}\zeta\left(v\right)h\left(\cdot,\cdot,v\right)\,dv\in H^{\frac{1}{4}}\left(\R\times\R^3\right)
	\end{align*}
	with
	\begin{align*}
	\left\|\int_{B_r}\zeta\left(v\right)h\left(\cdot,\cdot,v\right)\,dv\right\|_{H^{\frac{1}{4}}\left(\R\times\R^3\right)}\leq C\left(\left\|h\right\|_{L^2\left(\R\times\R^3\times B_r\right)}+\left\|g_0\right\|_{L^2\left(\R\times\R^3\times B_r\right)}+\left\|g_1\right\|_{L^2\left(\R\times\R^3\times B_r\right)}\right).
	\end{align*}
\end{lemma}
Here and in the following, $B_r\subset\R^3$ denotes the open ball about the origin with radius $r>0$.

We proceed with the following lemma, that was already mentioned above:
\begin{lemma}\label{lma:etafinL2H-1}
	Let $\f\in L^\infty\left(\left[0,T\right]\times\Omega\times\R^3\right)$, $\f_+\in L_\loc^1\left(\gamma_{T}^+,d\gamm\right)$ such that \cref{def:WeakSolWholeSys} \cref{def:WeakSolWholeSysii} is satisfied with $E,H\in L^2\left(\left[0,T\right]\times\Omega;\R^3\right)$. Denote $F\coloneqq\e\left(E+\v\times H\right)$. Then,
	\begin{align}\label{eq:etafinL2H-1}
	\partial_t\f+\v\cdot\partial_x\f=-\div_v\left(F\f\right)
	\end{align}
	on $\left]0,T\right[\times\Omega\times\R^3$ in the sense of distributions and the left-hand side is an element of \linebreak $L^2\left(\left[0,T\right]\times\Omega;H^{-1}\left(\R^3\right)\right)$. Furthermore,
	\begin{align}\label{eq:estNa}
	\Na\left(\f\right)\leq \sqrt{2}\left|\e\right|\left\|\f\right\|_{L^\infty\left(\left[0,T\right]\times\Omega\times\R^3\right)}\left\|\left(E,H\right)\right\|_{L^2\left(\left[0,T\right]\times\Omega;\R^6\right)}.
	\end{align}
\end{lemma}
\begin{proof}
	It is easy to see that \cref{eq:etafinL2H-1} holds on $\left]0,T\right[\times\Omega\times\R^3$ in the sense of distributions. There remains to estimate the right-hand side:
	\begin{align*}
	\left\|F\f\right\|_{L^2\left(\left[0,T\right]\times\Omega\times\R^3;\R^3\right)}\leq\left\|F\right\|_{L^2\left(\left[0,T\right]\times\Omega;\R^3\right)}\left\|\f\right\|_{L^\infty\left(\left[0,T\right]\times\Omega\times\R^3\right)}
	\end{align*}
	implies \cref{eq:estNa}.
\end{proof}
The next lemma gives an $L^{\frac{4}{3}}$-estimate on $j^\inte$ in view of the inequality constraints of \cref{prob:P} and will be useful later.
\begin{lemma}\label{lma:jintconstr}
	The constraints \cref{eq:constrfinfty,eq:constrener} yield $j^\inte\in L^{\frac{4}{3}}\left(\left[0,T\right]\times\Omega;\R^3\right)$ with
	\begin{align*}
	\left\|j^\inte\right\|_{L^{\frac{4}{3}}\left(\left[0,T\right]\times\Omega;\R^3\right)}\leq\left(\suma\left|\e\right|^4\left(\frac{4\pi}{3}\left\|\mathring\f\right\|_{L^\infty\left(\Omega\times\R^3\right)}+1\right)^4\right)^{\frac{1}{4}}\mathcal I\left(u\right)^{\frac{3}{4}}.
	\end{align*}
\end{lemma}
\begin{proof}
	The proof is standard, but we carry it out for the sake of completeness. For any $r>0$ we have
	\begin{align*}
	\int_{\R^3}\f\,dv=\int_{B_r}\f\,dv+\int_{\left|v\right|\geq r}\f\,dv\leq\frac{4\pi}{3}r^3\left\|\f\right\|_{L^\infty\left(\left[0,T\right]\times\Omega\times\R^3\right)}+\frac{1}{r}\int_{\R^3}\vo\f\,dv.
	\end{align*}
	Choosing $r\coloneqq\left(\int_{\R^3}\vo f\,dv\right)^{\frac{1}{4}}$ yields
	\begin{align}\label{eq:rhopoint}
	\int_{\R^3}\f\,dv&\leq\left(\int_{\R^3}\vo\f\,dv\right)^{\frac{3}{4}}\left(\frac{4\pi}{3}\left\|\f\right\|_{L^\infty\left(\left[0,T\right]\times\Omega\times\R^3\right)}+1\right)\nonumber\\
	&\leq\left(\int_{\R^3}\vo\f\,dv\right)^{\frac{3}{4}}\left(\frac{4\pi}{3}\left\|\mathring\f\right\|_{L^\infty\left(\Omega\times\R^3\right)}+1\right)
	\end{align}
	by \cref{eq:constrfinfty}, whence
	\begin{align}\label{eq:jintest43wholetime}
	&\left(\int_0^{T}\int_\Omega\left|j^\inte\right|^{\frac{4}{3}}\,dxdt\right)^{\frac{3}{4}}\leq\suma\left|\e\right|\left(\int_0^{T}\int_\Omega\left|\int_{\R^3}\f\,dv\right|^{\frac{4}{3}}\,dxdt\right)^{\frac{3}{4}}\nonumber\\
	&\leq\left(\suma\left|\e\right|^4\left(\frac{4\pi}{3}\left\|\mathring\f\right\|_{L^\infty\left(\Omega\times\R^3\right)}+1\right)^4\right)^{\frac{1}{4}}\left(\suma\int_0^{T}\int_\Omega\int_{\R^3}\vo\f\,dvdxdt\right)^{\frac{3}{4}},
	\end{align}
	which, together with the constraint \cref{eq:constrener}, implies the assertion.
\end{proof}
We can now prove the following:
\begin{theorem}\label{thm:ExistenceMinimizer}
	There is a (not necessarily unique) minimizer of \cref{prob:P}.
\end{theorem}
\begin{proof}
	First notice that there are feasible points to \cref{prob:P} by \cref{prop:Existence}. Thus, we may consider a minimizing sequence $\left(\left(\f_k,\f_{k,+}\right)_\alpha,E_k,H_k,u_k\right)$ of \cref{prob:P}. By structure of $\mathcal J$, the sequences $\left(\f_{k,+}\right)$ are bounded in $L^q\left(\gamma_{T}^+,d\gamm\right)$ and the sequence $\left(u_k\right)$ is bounded in $\mathcal U$. By reflexivity, we may thus assume that these sequences converge weakly, after possibly extracting suitable subsequences, in the respective spaces to some $\f_+\in L^q\left(\gamma_{T}^+,d\gamm\right)$ and $u\in\mathcal U$; recall that $1<q<\infty$.
	
	Since $\mathcal U$ is compactly embedded in $L^2\left(\left[0,T\right]\times\Gamma;\R^3\right)$, we have
	\begin{align}\label{eq:exminukstrong}
	\left\|u\right\|_{L^2\left(\left[0,T\right]\times\Gamma;\R^3\right)}=\lim_{k\to\infty}\left\|u_k\right\|_{L^2\left(\left[0,T\right]\times\Gamma;\R^3\right)}.
	\end{align}
	In combination with the constraints \cref{eq:constrfinfty,eq:constrener}, this yields that the sequences $\left(\f_k\right)$ are bounded in $\left(L_\kin^1\cap L^\infty\right)\left(\left[0,T\right]\times\Omega\times\R^3\right)$ and that the sequence $\left(E_k,H_k\right)$ is bounded in $L^2\left(\left[0,T\right]\times\R^3;\R^6\right)$. The property of $\left(\f_k\right)$ implies the boundedness of $\left(\f_k\right)$ in any $L^p\left(\left[0,T\right]\times\Omega\times\R^3\right)$, $1\leq p\leq\infty$, by interpolation. Therefore, after extracting a further subsequence, $\f_k$ converges weakly to some $\f$ in any $L^p\left(\left[0,T\right]\times\Omega\times\R^3\right)$, $1<p\leq\infty$ (weakly-* if $p=\infty$), and $\left(E_k,H_k\right)$ converges weakly to some $\left(E,H\right)$ in $L^2\left(\left[0,T\right]\times\R^3;\R^6\right)$.
	
	By weak-* convergence in $L^\infty\left(\left[0,T\right]\times\Omega\times\R^3\right)$, the constraint \cref{eq:constrfinfty} is preserved in the limit. As for the constraint \cref{eq:constrener}, let $R>0$. By weak convergence of the $\f_k$, weak convergence of $\left(E_k,H_k\right)$, \cref{eq:constrener} along the minimizing sequence, and \cref{eq:exminukstrong}, we have
	\begin{align*}
	&\suma\int_0^{T}\int_\Omega\int_{B_R}\vo\f\,dvdxdt+\frac{\sigma}{8\pi}\int_0^{T}\int_{\R^3}\left(\left|E\right|^2+\left|H\right|^2\right)\,dxdt\\
	&\leq\liminf_{k\to\infty}\suma\int_0^{T}\int_\Omega\int_{B_R}\vo\f_k\,dvdxdt+\frac{\sigma}{8\pi}\int_0^{T}\int_{\R^3}\left(\left|E_k\right|^2+\left|H_k\right|^2\right)\,dxdt\\
	&\leq 2T\suma\int_\Omega\int_{\R^3}\vo\mathring\f\,dvdx+\frac{T\sigma'}{4\pi}\left\|\left(\mathring E,\mathring H\right)\right\|_{L^2\left(\R^3;\R^6\right)}^2+2\pi T^2\sigma^{-1}\lim_{k\to\infty}\left\|u_k\right\|_{L^2\left(\left[0,T\right]\times\Gamma;\R^3\right)}^2\\
	&=2T\suma\int_\Omega\int_{\R^3}\vo\mathring\f\,dvdx+\frac{T\sigma'}{4\pi}\left\|\left(\mathring E,\mathring H\right)\right\|_{L^2\left(\R^3;\R^6\right)}^2+2\pi T^2\sigma^{-1}\left\|u\right\|_{L^2\left(\left[0,T\right]\times\Gamma;\R^3\right)}^2,
	\end{align*}
	which, after letting $R\to\infty$, on the one hand yields $\f\in L_\kin^1\left(\left[0,T\right]\times\Omega\times\R^3\right)$ and on the other hand implies that the constraint \cref{eq:constrener} also holds in the limit. Here we should point out that \cref{eq:exminukstrong} was crucial since we needed
	\begin{align*}
	\liminf_{k\to\infty}\left\|u_k\right\|_{L^2\left(\left[0,T\right]\times\Gamma;\R^3\right)}\leq\left\|u\right\|_{L^2\left(\left[0,T\right]\times\Gamma;\R^3\right)}.
	\end{align*}
	If we had chosen a cost term with the $L^2\left(\left[0,T\right]\times\Gamma;\R^3\right)$-norm instead of the $\mathcal U$-norm of $u$ in the objective function, we would only have been able to extract a subsequence $\left(u_k\right)$ that converges weakly in $L^2\left(\left[0,T\right]\times\Gamma;\R^3\right)$ rendering the above $\liminf$-estimate false in general.
	
	The next step is to pass to the limit in the weak form of \eqref{eq:WholeSystem}. By \cref{lma:jintconstr} the internal currents converge weakly, after extracting a further subsequence, in $L^{\frac{4}{3}}\left(\left[0,T\right]\times\Omega;\R^3\right)$. The weak limit, call it $\tilde j^\inte$, has to be the internal current $j^\inte$ induced by the limit functions $\f$ because of the following: Take $\vartheta\in C_c^\infty\left(\left]0,T\right[\times\Omega;\R^3\right)$ and $r>0$. Using weak convergence of $j_k^\inte$ and $\f_k$, respectively, we deduce
	\begin{align*}
	&\left|\int_0^{T}\int_\Omega\left(j^\inte-\tilde j^\inte\right)\cdot\vartheta\,dxdt\right|=\left|\lim_{k\to\infty}\iint\limits_{\supp\vartheta}\left(j^\inte-j_k^\inte\right)\cdot\vartheta\,dxdt\right|\\
	&=\left|\lim_{k\to\infty}\iint\limits_{\supp\vartheta}\left(\suma\e\int_{\R^3}\v\f\,dv-\suma\e\int_{\R^3}\v\f_k\,dv\right)\cdot\vartheta\,dxdt\right|\\
	&\leq\limsup_{k\to\infty}\left|\suma\e\iint\limits_{\supp\vartheta}\int_{B_r}\v\left(\f-\f_k\right)\,dvdxdt\right|\nonumber\\
	&\phantom{=\;}+\left|\suma\e\iint\limits_{\supp\vartheta}\left(\int_{\left|v\right|\geq r}\v\f\,dv-\int_{\left|v\right|\geq r}\v\f_k\,dv\right)\cdot\vartheta\,dxdt\right|\\
	&\leq 0+\limsup_{m\to\infty}\frac{1}{r}\left\|\vartheta\right\|_\infty\suma\left|\e\right|\iint\limits_{\supp\vartheta}\left(\int_{\R^3}\vo\f\,dv+\int_{\R^3}\vo\f_k\,dv\right)\,dxdt\\
	&\leq\frac{C}{r},
	\end{align*}
	where $C$ is finite by virtue of \cref{eq:constrener} and the boundedness of $\left(u_k\right)$, and does not depend on $r$. Since $r>0$ and $\vartheta\in C_c^\infty\left(\left]0,T\right[\times\Omega;\R^3\right)$ was arbitrary, we conclude $j^\inte=\tilde j^\inte$. Thus, we can pass to the limit in \cref{eq:Vlasovweak,eq:Maxwellweak} easily in all terms but the nonlinear one. To handle this remaining term, we apply \cref{lma:MomentumAveraging} in a well known way. We carry out this application in order to explain the necessity of imposing the constraints \cref{eq:constrfinfty,eq:constrener}: Let $\zeta\in C_c^\infty\left(\R^3\right)$ and $r>0$ such that $\zeta$ vanishes for $\left|v\right|>r-1$. Our goal is to show that $\int_{\R^3}\zeta\f_k\,dv$ converges strongly (and not only weakly) to $\int_{\R^3}\zeta\f\,dv$ in $L^2\left(\left[0,T\right]\times\Omega\right)$. To this end, let $\eta\in C_c^\infty\left(\left]0,T\right[\times\Omega\times B_r\right)$. We have
	\begin{align*}
	&\partial_t\left(\eta\f_k\right)+\v\cdot\partial_x\left(\eta\f_k\right)\\
	&=-\div_v\left(\e\left(E_k+\v\times H_k\right)\left(\eta\f_k\right)\right)+\f_k\partial_t\eta+\f_k\v\cdot\partial_x\eta+\e\f_k\left(E_k+\v\times H_k\right)\cdot\partial_v\eta\\
	&\eqqcolon\div_vg_1^k+g_0^k
	\end{align*}
	on whole $\R\times\R^3\times\R^3$ in the sense of distributions. Clearly, the $L^2\left(\R\times\R\times B_r\right)$-norms of $g_0^k$ and $g_1^k$ are uniformly bounded in $k$ due to $\eta\in C_c^\infty\left(\left]0,T\right[\times\Omega\times B_r\right)$ and the already known uniform boundedness of $\f_k$ in $L^\infty$ and $L^2$ and $E_k,H_k$ in $L^2$ -- the latter being a consequence of imposing \cref{eq:constrfinfty,eq:constrener}! Thus, applying \cref{lma:MomentumAveraging} yields the uniform boundedness of
	\begin{align*}
	\left\|\int_{B_r}\zeta\left(v\right)\left(\eta\f_k\right)\left(\cdot,\cdot,v\right)\,dv\right\|_{H^{\frac{1}{4}}\left(\R\times\R^3\right)}=\left\|\int_{B_r}\zeta\left(v\right)\left(\eta\f_k\right)\left(\cdot,\cdot,v\right)\,dv\right\|_{H^{\frac{1}{4}}\left(\left[0,T\right]\times\Omega\right)}.
	\end{align*}
	By boundedness of $\left[0,T\right]\times\Omega$, $H^{\frac{1}{4}}\left(\left[0,T\right]\times\Omega\right)$ is compactly embedded in $L^2\left(\left[0,T\right]\times\Omega\right)$ so that the sequence $\left(\int_{B_r}\zeta\left(v\right)\left(\eta\f_k\right)\left(\cdot,\cdot,v\right)\,dv\right)$ converges, after extracting a suitable subsequence, strongly to $\int_{B_r}\zeta\left(v\right)\left(\eta\f\right)\left(\cdot,\cdot,v\right)\,dv$ in $L^2\left(\left[0,T\right]\times\Omega\right)$. Again by the uniform boundedness of the $\f_k$ in $L^2$ we can estimate
	\begin{align}\label{eq:etaest1}
	&\left\|\int_{\R^3}\zeta\left(v\right)\left(\left(1-\eta\right)\left(\f_k-\f\right)\right)\left(\cdot,\cdot,v\right)\,dv\right\|_{L^2\left(\left[0,T\right]\times\Omega\right)}\nonumber\\
	&=\left\|\int_{B_r}\zeta\left(v\right)\left(\left(1-\eta\right)\left(\f_k-\f\right)\right)\left(\cdot,\cdot,v\right)\,dv\right\|_{L^2\left(\left[0,T\right]\times\Omega\right)}\leq C\left\|1-\eta\right\|_{L^2\left(\left[0,T\right]\times\Omega\times B_r\right)}
	\end{align}
	with a constant $C\geq 0$ that does not depend on $k$. Now let $\iota>0$ be arbitrary (here and throughout this paper, the letter $\iota$, and not $\varepsilon$, will always denote a small positive number, since $\varepsilon$ is already used for the permittivity) and choose $\eta\in C_c^\infty\left(\left]0,T\right[\times\Omega\times B_r\right)$ such that the right-hand side of \cref{eq:etaest1} is smaller than $\iota$ -- note that $\left[0,T\right]\times\Omega\times B_r$ is bounded. For this fixed $\eta$, there holds
	\begin{align}\label{eq:etaest2}
	\left\|\int_{\R^3}\zeta\left(v\right)\left(\eta\left(\f_k-\f\right)\right)\left(\cdot,\cdot,v\right)\,dv\right\|_{L^2\left(\left[0,T\right]\times\Omega\right)}=\left\|\int_{B_r}\zeta\left(v\right)\left(\eta\left(\f_k-\f\right)\right)\left(\cdot,\cdot,v\right)\,dv\right\|_{L^2\left(\left[0,T\right]\times\Omega\right)}<\iota
	\end{align}
	for $k$ large enough by the strong convergence obtained above. Adding \cref{eq:etaest1,eq:etaest2} yields
	\begin{align*}
	\left\|\int_{\R^3}\zeta\left(v\right)\left(\f_k-\f\right)\left(\cdot,\cdot,v\right)\,dv\right\|_{L^2\left(\left[0,T\right]\times\Omega\right)}<2\iota
	\end{align*}
	for $k$ large enough, which implies
	\begin{align}\label{eq:Iterstrongconv}
	\int_{\R^3}\zeta\left(v\right)\f_k\left(\cdot,\cdot,v\right)\,dv\to\int_{\R^3}\zeta\left(v\right)\f\left(\cdot,\cdot,v\right)\,dv\mathrm{\ strongly\ in\ }L^2\left(\left[0,T\right]\times\Omega\right)\mathrm{\ for\ }k\to\infty.
	\end{align}
	Finally take $\psi\in\Psi_{T}$ and consider the limit of the crucial product term in \cref{eq:Vlasovweak}. By a density argument -- in particular, Weierstraß' approximation theorem -- we may assume that $\psi$ factorizes, i.e.,
	\begin{align*}
	\psi\left(t,x,v\right)=\psi_1\left(t,x\right)\psi_2\left(v\right).
	\end{align*}
	We have
	\begin{align*}
	&\lim_{k\to\infty}\int_0^{T}\int_\Omega\int_{\R^3}E_k\cdot\left(\partial_v\psi\right)\f_k\,dvdxdt=\lim_{k\to\infty}\int_0^{T}\int_\Omega E_k\psi_1\cdot
	\left(\int_{\R^3}\psi_2'\f_k\,dv\right)dxdt\nonumber\\
	&=\int_0^{T}\int_\Omega E\psi_1\cdot\left(
	\int_{\R^3}\psi_2'\f\,dv\right)dxdt=\int_0^{T}\int_\Omega\int_{\R^3} E\cdot\left(\partial_v\psi\right)\f\,dvdxdt
	\end{align*}
	by $\psi_1\in L^\infty\left(\left[0,T\right]\times\Omega\right)$, $E_k\wto E$ weakly in $L^2\left(\left[0,T\right]\times\Omega\right)$, and \cref{eq:Iterstrongconv} defining $\zeta\coloneqq\left(\psi_2'\right)_i$, $i=1,2,3$. Similarly, we obtain 
	\begin{align*}
	&\lim_{k\to\infty}\int_0^{T}\int_\Omega\int_{\R^3}\left(\v\times H_k\right)\cdot\left(\partial_v\psi\right)\f_k\,dvdxdt=\int_0^{T}\int_\Omega\int_{\R^3}\left(\v\times H\right)\cdot\left(\partial_v\psi\right)\f\,dvdxdt.
	\end{align*}
	Altogether, \eqref{eq:WholeSystem} is satisfied in the limit.
	
	By \cref{lma:etafinL2H-1}, we even have $\f\in\Y$ and thus $y=\left(\left(\f,\f_+\right)_\alpha,E,H\right)\in\mathcal Y$ altogether.
	
	Finally, the objective function indeed admits its minimum at $\left(y,u\right)$ by weak lower semi-continuity of any norm.
\end{proof}

\section{Weak formulation - revisited}\label{sec:weakrevisited}
For later reasons, it is convenient to revisit the weak formulation of \cref{def:WeakSolWholeSys} and write the equations there as an identity
\begin{align*}
G\left(\left(\f,\f_+\right)_\alpha,E,H,j\right)=0
\end{align*}
in the dual space of some reflexive space. Throughout this section, we fix $1<p<2$, $2<p',q<\infty$ such that
\begin{subequations}\label{eq:exponents}
	\begin{align}
	\frac{1}{p}+\frac{1}{q}&=1,\\
	\frac{1}{p'}+\frac{1}{q}+\frac{1}{2}&=1.
	\end{align}
\end{subequations}
We assume $\f\in L^q\left(\left[0,T\right]\times\Omega\times\R^3\right)$, $\f_+\in L^q\left(\gamma_{T}^+,d\gamm\right)$, $E,H\in L^2\left(\left[0,T\right]\times\R^3;\R^3\right)$, and $j=j^\inte+u$ where $j^\inte=\suma\e\int_{\R^3}\v\f\,dv\in L^{\frac{4}{3}}\left(\left[0,T\right]\times\Omega;\R^3\right)$, $u\in L^2\left(\left[0,T\right]\times\Gamma;\R^3\right)$. Note that for such $u$ there is a solution in the sense of \cref{def:WeakSolWholeSys} with those properties due to \cref{prop:Existence}.

Clearly, \cref{def:WeakSolWholeSys} \cref{def:WeakSolWholeSysii,def:WeakSolWholeSysiii} are equivalent to
\begin{align}
0&=\suma\left(-\int_0^{T}\int_\Omega\int_{\R^3}\left(\partial_t\psia+\v\cdot\partial_x\psia+\e\left(E+\v\times H\right)\cdot\partial_v\psia\right)\f\,dvdxdt\right.\nonumber\\
&\phantom{=\;}\left.+\int_{\gamma_{T}^+}\f_+\psia\,d\gamm-\int_{\gamma_{T}^-}\a\left(K\f_+\right)\psia\,d\gamm-\int_\Omega\int_{\R^3}\mathring\f\psia\left(0\right)\,dvdx\right)\nonumber\\
&\phantom{=\;}+\int_0^{T}\int_{\R^3}\left(\varepsilon E\cdot\partial_t\vartheta^e-H\cdot\curl_x\vartheta^e-4\pi j\cdot\vartheta^e\right)\,dxdt+\int_{\R^3}\varepsilon\mathring{E}\cdot\vartheta^e\left(0\right)\,dx\nonumber\\
&\phantom{=\;}+\int_0^{T}\int_{\R^3}\left(\mu H\cdot\partial_t\vartheta^h+E\cdot\curl_x\vartheta^h\right)\,dxdt+\int_{\R^3}\mu\mathring{H}\cdot\vartheta^h\left(0\right)\,dx\nonumber\\
&\eqqcolon G\left(\left(\f,\f_+\right)_\alpha,E,H,j\right)\left(\left(\psia\right)_\alpha,\vartheta^e,\vartheta^h\right)\label{eq:defS}
\end{align}
for all $\left(\psia\right)_\alpha\in\Psi_{T}^N$ and $\vartheta^e,\vartheta^h\in\Theta_{T}$.

\subsection{Some estimates}\label{sec:somestimates}
From now on, $\left(\left(\f,\f_+\right)_\alpha,E,H,j\right)$ does not have to be a solution of \eqref{eq:WholeSystem}. All assertions are made under the assumptions mentioned above.

In the following we will estimate each summand, one by one, where we often need \cref{eq:exponents}. Furthermore, $C$ denotes various positive, finite constants that only depend on $T$, $\Omega$, and $\Gamma$.
\begin{align*}
&\left|\int_0^{T}\int_\Omega\int_{\R^3}\partial_t\psia\f\,dvdxdt\right|\\
&\leq\sqrt{T\lambda\left(\Omega\right)}\int_{\R^3}\left(\int_0^{T}\int_\Omega\left|\f\right|^q\,dxdt\right)^{\frac{1}{q}}\left(\int_0^{T}\int_\Omega\left|\partial_t\psia\right|^{p'}\,dxdt\right)^{\frac{1}{p'}}\,dv\\
&\leq C\left\|\f\right\|_{L^q\left(\left[0,T\right]\times\Omega\times\R^3\right)}\left(\int_{\R^3}\left(\int_0^{T}\int_\Omega\left|\partial_t\psia\right|^{p'}\,dxdt\right)^{\frac{p}{p'}}\,dv\right)^{\frac{1}{p}};
\end{align*}
next
\begin{align*}
&\left|\int_0^{T}\int_\Omega\int_{\R^3}\v\cdot\partial_x\psia\f\,dvdxdt\right|\\
&\leq\sqrt{T\lambda\left(\Omega\right)}\int_{\R^3}\left(\int_0^{T}\int_\Omega\left|\f\right|^q\,dxdt\right)^{\frac{1}{q}}\left(\int_0^{T}\int_\Omega\left|\partial_x\psia\right|^{p'}\,dxdt\right)^{\frac{1}{p'}}\,dv\\
&\leq C\left\|\f\right\|_{L^q\left(\left[0,T\right]\times\Omega\times\R^3\right)}\left(\int_{\R^3}\left(\int_0^{T}\int_\Omega\left|\partial_x\psia\right|^{p'}\,dxdt\right)^{\frac{p}{p'}}\,dv\right)^{\frac{1}{p}};
\end{align*}
then,
\begin{align*}
&\left|\int_0^{T}\int_\Omega\int_{\R^3}\e\left(E+\v\times H\right)\cdot\partial_v\psia\f\,dvdxdt\right|\\
&\leq\left|\e\right|\int_{\R^3}\left(\int_0^{T}\int_\Omega\left|\f\right|^q\,dxdt\right)^{\frac{1}{q}}\left(\int_0^{T}\int_\Omega\left|E+\v\times H\right|^2\,dxdt\right)^{\frac{1}{2}}\left(\int_0^{T}\int_\Omega\left|\partial_v\psia\right|^{p'}\,dxdt\right)^{\frac{1}{p'}}\,dv\\
&\leq 2\left|\e\right|\left\|\f\right\|_{L^q\left(\left[0,T\right]\times\Omega\times\R^3\right)}\left\|\left(E,H\right)\right\|_{L^2\left(\left[0,T\right]\times\Omega;\R^6\right)}\left(\int_{\R^3}\left(\int_0^{T}\int_\Omega\left|\partial_v\psia\right|^{p'}\,dxdt\right)^{\frac{p}{p'}}\,dv\right)^{\frac{1}{p}}.
\end{align*}
Now have in mind that there is a bounded trace operator 
\begin{align*}
W^{1,p'}\left(\left[0,T\right]\times\Omega\right)\to L^{p'}\left(\left(\left[0,T\right]\times\partial\Omega\right)\cup\left(\left\{0\right\}\times\Omega\right)\cup\left(\left\{T\right\}\times\Omega\right)\right).
\end{align*}
Thus,
\begin{align*}
&\left|\int_{\gamma_{T}^+}\f_+\psia\,d\gamm\right|\leq\int_{\R^3}\int_0^{T}\int\limits_{\left\{x\in\partial\Omega\mid\v\cdot n\left(x\right)>0\right\}}\left|\f_+\psia\right|\left|\v\cdot n\right|\,dS_xdtdv\\
&\leq C\int_{\R^3}\left(\int_0^{T}\int\limits_{\left\{x\in\partial\Omega\mid\v\cdot n\left(x\right)>0\right\}}\left|\f_+\right|^q\left|\v\cdot n\right|^q\,dS_xdt\right)^{\frac{1}{q}}\left(\int_0^{T}\int_{\partial\Omega}\left|\psia\right|^{p'}\,dS_xdt\right)^{\frac{1}{p'}}dv\\
&\leq C\int_{\R^3}\left(\int_0^{T}\int\limits_{\left\{x\in\partial\Omega\mid\v\cdot n\left(x\right)>0\right\}}\left|\f_+\right|^q\left|\v\cdot n\right|\,dS_xdt\right)^{\frac{1}{q}}\\
&\phantom{\leq C\int_{\R^3}}\cdot\left(\int_0^{T}\int_{\Omega}\left(\left|\psia\right|^{p'}+\left|\partial_t\psia\right|^{p'}+\left|\partial_x\psia\right|^{p'}\right)\,dxdt\right)^{\frac{1}{p'}}dv\\
&\leq C\left\|\f_+\right\|_{L^q\left(\gamma_{T}^+,d\gamm\right)}\left(\int_{\R^3}\left(\int_0^{T}\int_\Omega\left(\left|\psia\right|^{p'}+\left|\partial_t\psia\right|^{p'}+\left|\partial_x\psia\right|^{p'}\right)\,dxdt\right)^{\frac{p}{p'}}\,dv\right)^{\frac{1}{p}}
\end{align*}
by $\left|\v\cdot n\right|\leq 1$. Similarly,
\begin{align*}
&\left|\int_{\gamma_{T}^-}\a\left(K\f_+\right)\psia\,d\gamm\right|\\
&\leq C\left\|\f_+\right\|_{L^q\left(\gamma_{T}^+,d\gamm\right)}\left(\int_{\R^3}\left(\int_0^{T}\int_\Omega\left(\left|\psia\right|^{p'}+\left|\partial_t\psia\right|^{p'}+\left|\partial_x\psia\right|^{p'}\right)\,dxdt\right)^{\frac{p}{p'}}\,dv\right)^{\frac{1}{p}}
\end{align*}
since $\left|\a\right|\leq 1$ and $v\mapsto v-2\left(v\cdot n\left(x\right)\right)n\left(x\right)$ has Jacobian determinant $-1$. Analogously,
\begin{align*}
&\left|\int_\Omega\int_{\R^3}\mathring\f\psia\left(0\right)\,dvdx\right|\\
&\leq C\left\|\mathring\f\right\|_{L^q\left(\Omega\times\R^3\right)}\left(\int_{\R^3}\left(\int_0^{T}\int_\Omega\left(\left|\psia\right|^{p'}+\left|\partial_t\psia\right|^{p'}+\left|\partial_x\psia\right|^{p'}\right)\,dxdt\right)^{\frac{p}{p'}}\,dv\right)^{\frac{1}{p}}
\end{align*}
making use of the boundedness of the trace operator, now regarding the slice $\left\{0\right\}\times\Omega$ instead of $\left[0,T\right]\times\partial\Omega$.

As for the Maxwell part, we can easily estimate
\begin{align*}
\left|\int_0^{T}\int_{\R^3}\varepsilon E\cdot\partial_t\vartheta^e\,dxdt\right|&\leq\sigma'\left\|E\right\|_{L^2\left(\left[0,T\right]\times\R^3;\R^3\right)}\left\|\partial_t\vartheta^e\right\|_{L^2\left(\left[0,T\right]\times\R^3;\R^3\right)},\\
\left|\int_0^{T}\int_{\R^3}H\cdot\curl_x\vartheta^e\,dxdt\right|&\leq \left\|H\right\|_{L^2\left(\left[0,T\right]\times\R^3;\R^3\right)}\left\|\curl_x\vartheta^e\right\|_{L^2\left(\left[0,T\right]\times\R^3;\R^3\right)},\\
\left|\int_0^{T}\int_{\R^3}\mu H\cdot\partial_t\vartheta^h\,dxdt\right|&\leq\sigma'\left\|H\right\|_{L^2\left(\left[0,T\right]\times\R^3;\R^3\right)}\left\|\partial_t\vartheta^h\right\|_{L^2\left(\left[0,T\right]\times\R^3;\R^3\right)},\\
\left|\int_0^{T}\int_{\R^3}E\cdot\curl_x\vartheta^h\,dxdt\right|&\leq \left\|E\right\|_{L^2\left(\left[0,T\right]\times\R^3;\R^3\right)}\left\|\curl_x\vartheta^h\right\|_{L^2\left(\left[0,T\right]\times\R^3;\R^3\right)}.
\end{align*}
Concerning the terms with the initial data, we first notice that for all $x\in\R^3$ we have
\begin{align*}
\vartheta^e\left(0,x\right)=\vartheta^e\left(0,x\right)-\vartheta^e\left(T,x\right)=-\int_0^{T}\partial_t\vartheta^e\left(t,x\right)\,dt,
\end{align*}
consequently
\begin{align*}
\left|\vartheta^e\left(0,x\right)\right|^2\leq T\int_0^{T}\left|\partial_t\vartheta^e\left(t,x\right)\right|^2\,dt
\end{align*}
and therefore
\begin{align*}
\left|\int_{\R^3}\varepsilon\mathring E\cdot\vartheta^e\left(0\right)\,dx\right|\leq\sigma'C\left\|\mathring E\right\|_{L^2\left(\R^3;\R^3\right)}\left\|\partial_t\vartheta^e\right\|_{L^2\left(\left[0,T\right]\times\R^3;\R^3\right)}.
\end{align*}
Similarly we conclude
\begin{align*}
\left|\int_{\R^3}\mu\mathring H\cdot\vartheta^h\left(0\right)\,dx\right|\leq\sigma'C\left\|\mathring H\right\|_{L^2\left(\R^3;\R^3\right)}\left\|\partial_t\vartheta^h\right\|_{L^2\left(\left[0,T\right]\times\R^3;\R^3\right)}.
\end{align*}
Lastly, we turn to the term with $j$. By Sobolev's embedding theorem, $H^1\left(\left[0,T\right]\times A\right)$ is continuously embedded in $L^4\left(\left[0,T\right]\times A\right)$, $A=\Omega,\Gamma$, yielding
\begin{align*}
&\left|\int_0^{T}\int_{\R^3}j\cdot\vartheta^e\,dxdt\right|\leq\left|\int_0^{T}\int_{\Omega}j^\inte\cdot\vartheta^e\,dxdt\right|+\left|\int_0^{T}\int_{\Gamma}u\cdot\vartheta^e\,dxdt\right|\\
&\leq\left\|j^\inte\right\|_{L^{\frac{4}{3}}\left(\left[0,T\right]\times\Omega;\R^3\right)}\left\|\vartheta^e\right\|_{L^4\left(\left[0,T\right]\times\Omega;\R^3\right)}+\left\|u\right\|_{L^2\left(\left[0,T\right]\times\Gamma;\R^3\right)}\left\|\vartheta^e\right\|_{L^2\left(\left[0,T\right]\times\Gamma;\R^3\right)}\\
&\leq C\left(\left\|j^\inte\right\|_{L^{\frac{4}{3}}\left(\left[0,T\right]\times\Omega;\R^3\right)}+\left\|u\right\|_{L^2\left(\left[0,T\right]\times\Gamma;\R^3\right)}\right)\left\|\vartheta^e\right\|_{H^1\left(\left[0,T\right]\times\R^3;\R^3\right)}.
\end{align*}
Altogether, we conclude that $G\left(\left(\f,\f_+\right)_\alpha,E,H,j\right)$ is a bounded linear operator on $\Psi_{T}^N\times\Theta_{T}^2$ if we equip $\Psi_{T}$ with the norm
\begin{align}\label{eq:normW1pp'}
\left\|\psi\right\|_{W^{1,p,p'}}\coloneqq\left(\int_{\R^3}\left(\int_0^{T}\int_\Omega\left(\left|\psi\right|^{p'}+\left|\partial_t\psi\right|^{p'}+\left|\partial_x\psi\right|^{p'}+\left|\partial_v\psi\right|^{p'}\right)\,dxdt\right)^{\frac{p}{p'}}\,dv\right)^{\frac{1}{p}}
\end{align}
and $\Theta_{T}$ with the usual $H^1$-norm on $\left[0,T\right]\times\R^3$.

\subsection{The space \texorpdfstring{$W^{1,p,p'}$}{W1,p,p'} and the extended functional}
The choice of the norm for the test functions $\psi$ suggests having a detailed look at the space $W^{1,p,p'}$. This space, so to say a mixed order Sobolev space, is defined to be the space consisting of all measurable functions on $\left[0,T\right]\times\Omega\times\R^3$ with values in $\R$ such that their derivatives of first order are locally integrable functions and additionally the right-hand side of \cref{eq:normW1pp'} is finite.

We first consider the corresponding $L^{p,p'}$-space, that is
\begin{align*}
L^{p,p'}&\coloneqq\left\{\psi\colon\left[0,T\right]\times\Omega\times\R^3\to\R\ \mathrm{measurable}\mid\vphantom{\left(\int_{\R^3}\left(\int_0^{T}\int_\Omega\left|\psi\right|^{p'}\,dxdt\right)^{\frac{p}{p'}}\,dv\right)^{\frac{1}{p}}}\right.\\
&\phantom{\coloneqq\ \Big\{\,}\left.\left\|\psi\right\|_{L^{p,p'}}\coloneqq\left(\int_{\R^3}\left(\int_0^{T}\int_\Omega\left|\psi\right|^{p'}\,dxdt\right)^{\frac{p}{p'}}\,dv\right)^{\frac{1}{p}}<\infty\right\}.
\end{align*}
Since we can identify $L^{p,p'}$ with the Bochner space $L^p\left(\R^3;L^{p'}\left(\left[0,T\right]\times\Omega\right)\right)$, we get the following basic property:
\begin{lemma}
	$L^{p,p'}$ is a uniformly convex Banach space.
\end{lemma}
\begin{proof}
	This is easy to see using the identification above. The uniform convexity follows from a classical result of Day \cite{Day41} since $1<p,p'<\infty$.
\end{proof}

The uniform convexity will be crucial later.

These properties of $L^{p,p'}$ carry over to $W^{1,p,p'}$ in the same natural way as such properties carry over from standard $L^p$-spaces to standard Sobolev spaces $W^{1,p}$: The space $W^{1,p,p'}$ can be interpreted as a closed subspace of $\left(L^{p,p'}\right)^7$ via the isometry
\begin{align*}
\psi\mapsto\left(\psi,\partial_t\psi,\partial_{x_1}\psi,\partial_{x_2}\psi,\partial_{x_3}\psi,\partial_{v_1}\psi,\partial_{v_2}\psi,\partial_{v_3}\psi\right).
\end{align*}
Hence, one can argue in the same way as in the standard context to prove:
\begin{lemma}
	$W^{1,p,p'}$ is a uniformly convex, reflexive Banach space.
\end{lemma}
\begin{proof}
	Note that uniform convexity and completeness imply reflexivity by the classical Milman-Pettis theorem, see for example \cite{Pet39}.
\end{proof}
Now we turn back to our weak formulation. Recall that we have proved 
\begin{align*}
G\left(\left(\f,\f_+\right)_\alpha,E,H,j\right)\in\left(\Psi_{T}^N\times\Theta_{T}^2\right)^*.
\end{align*}
If we denote $\Lambda\coloneqq\overline{\Psi_{T}}^N\times\overline{\Theta_{T}}^2$, where the closure is to be understood in $W^{1,p,p'}$ and $H^1\left(\left[0,T\right]\times\R^3;\R^3\right)$, respectively, we can extend $G\left(\left(\f,\f_+\right)_\alpha,E,H,j\right)$ uniquely to a bounded linear operator on $\Lambda$ and still the formula in \cref{eq:defS} applies. Since $H^1\left(\left[0,T\right]\times\R^3;\R^3\right)$ is also a uniformly convex, reflexive Banach space and since $\Lambda$ is a closed subspace, we instantly conclude:
\begin{lemma}\label{lma:Lambdauniconrefl}
	$\Lambda$, equipped with the norm
	\begin{align*}
	\left\|\left(\left(\psia\right)_\alpha,\vartheta^e,\vartheta^h\right)\right\|_{\Lambda}\coloneqq\left(\suma\left\|\psia\right\|_{W^{1,p,p'}}^2+\left\|\vartheta^e\right\|_{H^1\left(\left[0,T\right]\times\R^3;\R^3\right)}^2+\left\|\vartheta^h\right\|_{H^1\left(\left[0,T\right]\times\R^3;\R^3\right)}^2\right)^{\frac{1}{2}},
	\end{align*}
	is a uniformly convex, reflexive Banach space.
\end{lemma}
\begin{proof}
	By Clarkson \cite{Cla36}, a finite Cartesian product of uniformly convex spaces is again uniformly convex if one sums up the norms properly. Note that we have chosen the $2$-norm on $\R^{N+2}$ to sum up the particular norms (any other $\tilde p$-Norm, $1<\tilde p<\infty$, would work as well). Thus, $\Lambda$ is uniformly convex. Again by completeness of $\Lambda$ and the Milman-Pettis theorem, we conclude that $\Lambda$ is additionally reflexive.
\end{proof}

Thus, we can regard $G\left(\left(\f,\f_+\right)_\alpha,E,H,j\right)\in\Lambda^*$ as an element of the dual space of a uniformly convex, reflexive Banach space, and we have that, under the assumptions made in the beginning of \cref{sec:weakrevisited}, $G\left(\left(\f,\f_+\right)_\alpha,E,H,j\right)=0$ is equivalent to $\left(\left(\f,\f_+\right)_\alpha,E,H,j\right)$ being a weak solution of the Vlasov-Maxwell system \eqref{eq:WholeSystem} on the time interval $\left[0,T\right]$.

Notice that $\Lambda$ is a proper subspace of $\left(W^{1,p,p'}\right)^N\times\left(H^1\left(\left[0,T\right]\times\R^3;\R^3\right)\right)^2$ since $\psi\in\Psi_{T}$ and $\vartheta\in\Theta_{T}$ vanish for $t=T$.

Later, in \cref{sec:firstorder}, we want to derive first order optimality conditions for a minimizer of \cref{prob:P}. To this end, it will be helpful that $G$ ($\mathcal G$, to be more precise, see below) is differentiable in $\left(\left(\f,\f_+\right)_\alpha,E,H,u\right)$ with respect to a suitable norm; here and in the following, differentiability always means differentiability in the sense of Fréchet. As in the formulation of \cref{prob:P}, we restrict ourselves to $\left(\left(\f,\f_+\right)_\alpha,E,H,u\right)\in\mathcal Y\times\mathcal U$. Note that this yields $\f\in L^q\left(\left[0,T\right]\times\Omega\times\R^3\right)$ by interpolation and thus we can carry through the previous considerations of this section. We equip $\mathcal Y\times\mathcal U$ with the norm
\begin{align*}
\left\|\left(y,u\right)\right\|_{\mathcal Y\times\mathcal U}&=\left\|\left(\left(\f,\f_+\right)_\alpha,E,H,u\right)\right\|_{\mathcal Y\times\mathcal U}\\
&\coloneqq\suma\left(\left\|\f\right\|_{\Y}+\left\|\f_+\right\|_{L^q\left(\gamma_{T}^+,d\gamm\right)}\right)+\left\|\left(E,H\right)\right\|_{L^2\left(\left[0,T\right]\times\R^3;\R^6\right)}+\left\|u\right\|_{\mathcal U},
\end{align*}
where
\begin{align*}
\left\|f\right\|_{\Y}&\coloneqq\left\|f\right\|_{L_\kin^1\left(\left[0,T\right]\times\Omega\times\R^3\right)}+\left\|f\right\|_{L^\infty\left(\left[0,T\right]\times\Omega\times\R^3\right)}+\Na\left(f\right).
\end{align*}
The latter indeed is a norm on $\Y$ since $\Na$ is a semi-norm on $\Y$, as is easily seen. Note that the following lemma does not need the adding of $\Na$ as above; however, this will heavily be exploited in \cref{sec:firstorder}. 

\begin{lemma}\label{lma:Gdiff} The map
	\begin{align*}
	\mathcal G\colon\mathcal Y\times\mathcal U&\to\Lambda^*,\\
	\mathcal G\left(\left(\f,\f_+\right)_\alpha,E,H,u\right)&=G\left(\left(\f,\f_+\right)_\alpha,E,H,j^\inte+u\right)
	\end{align*}
	is differentiable and we have
	\begin{align}\label{eq:Gdiff}
	&\left(\mathcal G'\left(y,u\right)\left(\delta y,\delta u\right)\right)\left(\left(\psia\right)_\alpha,\vartheta^e,\vartheta^h\right)\nonumber\\
	&=\suma\left(-\int_0^{T}\int_\Omega\int_{\R^3}\left(\left(\partial_t\psia+\v\cdot\partial_x\psia+\e\left(E+\v\times H\right)\cdot\partial_v\psia\right)\delta\f\right.\right.\nonumber\\
	&\omit\hfill$\displaystyle\left.+\e\left(\delta E+\v\times\delta H\right)\f\cdot\partial_v\psia\right)\,dvdxdt$\nonumber\\
	&\phantom{=\suma\Big(}\left.+\int_{\gamma_{T}^+}\delta\f_+\psia\,d\gamm-\int_{\gamma_{T}^-}\a\left(K\delta\f_+\right)\psia\,d\gamm\right)\nonumber\\
	&\phantom{=\;}+\int_0^{T}\int_{\R^3}\left(\varepsilon\delta E\cdot\partial_t\vartheta^e-\delta H\cdot\curl_x\vartheta^e-4\pi\left(\delta j^\inte+\delta u\right)\cdot\vartheta^e\right)\,dxdt\nonumber\\
	&\phantom{=\;}+\int_0^{T}\int_{\R^3}\left(\mu\delta H\cdot\partial_t\vartheta^h+\delta E\cdot\curl_x\vartheta^h\right)\,dxdt,
	\end{align}
	where, in accordance with the previous notation,
	\begin{align*}
	\delta j^\inte=\suma\e\int_{\R^3}\v\delta\f\,dv.
	\end{align*}
\end{lemma}
\begin{proof}
	The canonical candidate for the linearization at a point $\left(y,u\right)$ in direction $\left(\delta y,\delta u\right)=\left(\left(\delta\f,\delta\f_+\right)_\alpha,\delta E,\delta H,\delta u\right)$ is $\mathcal G'\left(y,u\right)\left(\delta y,\delta u\right)$ as stated above. Recalling the estimates of \cref{sec:somestimates}, we see that $\mathcal G'\left(y,u\right)\left(\delta y,\delta u\right)\in\Lambda^*$ and
	\begin{align}\label{eq:estG'Lambda}
	&\left\|\mathcal G'\left(y,u\right)\left(\delta y,\delta u\right)\right\|_{\Lambda^*}\nonumber\\
	&\leq C\left(\suma\left(\left\|\delta\f\right\|_{L^q\left(\left[0,T\right]\times\Omega\times\R^3\right)}+\left\|\left(E,H\right)\right\|_{L^2\left(\left[0,T\right]\times\R^3;\R^6\right)}\left\|\delta\f\right\|_{L^q\left(\left[0,T\right]\times\Omega\times\R^3\right)}\right.\right.\nonumber\\
	&\omit\hfill$\displaystyle\left.+\left\|\f\right\|_{L^q\left(\left[0,T\right]\times\Omega\times\R^3\right)}\left\|\left(\delta E,\delta H\right)\right\|_{L^2\left(\left[0,T\right]\times\R^3;\R^6\right)}+\left\|\delta\f_+\right\|_{L^q\left(\gamma_{T}^+,d\gamm\right)}\right)$\nonumber\\
	&\phantom{\leq\;}\left.\vphantom{\suma}+\left\|\left(\delta E,\delta H\right)\right\|_{L^2\left(\left[0,T\right]\times\R^3;\R^6\right)}+\left\|\delta j^\inte\right\|_{L^{\frac{4}{3}}\left(\left[0,T\right]\times\Omega;\R^3\right)}+\left\|\delta u\right\|_{L^2\left(\left[0,T\right]\times\Gamma;\R^3\right)}\right),
	\end{align}
	where $C$ only depends on $T$, $\Omega$, $\Gamma$, $\sigma'$, and the $\e$.
	
	Similarly to \cref{eq:rhopoint,eq:jintest43wholetime}, we deduce
	\begin{align*}
	&\left\|\delta j^\inte\right\|_{L^{\frac{4}{3}}\left(\left[0,T\right]\times\Omega;\R^3\right)}\\
	&\leq\left(\suma\left|\e\right|^4\left(\frac{4\pi}{3}\left\|\delta\f\right\|_{L^\infty\left(\left[0,T\right]\times\Omega\times\R^3\right)}+1\right)^4\right)^{\frac{1}{4}}\suma\left\|\delta\f\right\|_{L_\kin^1\left(\left[0,T\right]\times\Omega\times\R^3\right)}^{\frac{3}{4}}.
	\end{align*}
	This and \cref{eq:estG'Lambda} yield that $\mathcal G'\left(y,u\right)\left(\delta y,\delta u\right)\to 0$ in $\Lambda^*$ when $\left(\delta y,\delta u\right)\to 0$ in $\mathcal Y\times\mathcal U$. Therefore, $\mathcal G'\left(y,u\right)\colon\mathcal Y\times\mathcal U\to\Lambda^*$ is a bounded linear map; linearity is of course easy to see.
	
	To show that $\mathcal G'\left(y,u\right)$ indeed is the derivative of $\mathcal G$ at $\left(y,u\right)$, we consider the remainder, which only contains terms that come from the nonlinearity in the Vlasov-Maxwell system:
	\begin{align*}
	&\left(\mathcal G\left(y+\delta y,u+\delta u\right)-\mathcal G\left(y,u\right)-\mathcal G'\left(y,u\right)\left(\delta y,\delta u\right)\right)\left(\left(\psia\right)_\alpha,\vartheta^e,\vartheta^h\right)\\
	&=-\suma\e\int_0^{T}\int_\Omega\int_{\R^3}\left(\delta E+\v\times\delta H\right)\cdot\partial_v\psia\delta\f\,dvdxdt.
	\end{align*}
	Hence, again using the corresponding estimate of \cref{sec:somestimates},
	\begin{align*}
	&\left\|\mathcal G\left(y+\delta y,u+\delta u\right)-\mathcal G\left(y,u\right)-\mathcal G'\left(y,u\right)\left(\delta y,\delta u\right)\right\|_{\Lambda^*}\\
	&\leq C\suma\left\|\delta\f\right\|_{L^q\left(\left[0,T\right]\times\Omega\times\R^3\right)}\left\|\left(\delta E,\delta H\right)\right\|_{L^2\left(\left[0,T\right]\times\Omega;\R^6\right)}=o\left(\left\|\left(\delta y,\delta u\right)\right\|_{\mathcal Y\times\mathcal U}\right)
	\end{align*}
	for $\left(\delta y,\delta u\right)\to 0$ in $\mathcal Y\times\mathcal U$, where $C$ only depends on $\sigma'$ and the $\e$. This completes the proof.
\end{proof}

\section{First order optimality conditions}\label{sec:firstorder}
A standard step during treating an optimization problem is to derive first order necessary optimality conditions. Typically, one exploits differentiability of the control-to-state operator. Unfortunately, we do not have such an operator on hand since we do not even have uniqueness of weak solutions for a fixed control $u$. Lions \cite{Lio85} introduced a way to tackle optimization problems having a PDE (system), that (possibly) admits multiple solutions, as a constraint. The main strategy therefore is to consider approximate optimization problems that no longer have the PDE (system) as a constraint but merely penalize points that do not solve this PDE (system). For such approximate problems, one can show that minimizers exist and derive first order optimality conditions. Then the penalization parameter is driven to $\infty$ and one hopes the PDE (system) to be solved in the limit, that is to say the limit of minimizers (in whatever sense) is a solution of the PDE (system), and moreover it is a minimizer of the original problem. Furthermore, one may show that passage to the limit in the approximate optimality conditions, in particular in the adjoint PDE (system), is possible, too.

We fix $q>2$ and $p,p'$ satisfying \cref{eq:exponents} so that the results of \cref{sec:weakrevisited} can be applied.

\subsection{An approximate optimization problem}
Following the outlined strategy, we introduce a penalization parameter $s>0$ (which will be driven to $\infty$ later) and consider the approximate problem
\begin{align}\tag{\mbox{P$_\text{s}$}}\label{prob:Ps}
\left.
\begin{aligned}
\min_{y\in\mathcal Y,u\in\mathcal U}\quad&\mathcal J_s\left(y,u\right)=\frac{1}{q}\suma \w\left\|\f_+\right\|_{L^q\left(\gamma_{T}^+,d\gamm\right)}^q+\frac{1}{r}\left\|u\right\|_{\mathcal U}^r+\frac{s}{2}\left\|\mathcal G\left(y,u\right)\right\|_{\Lambda^*}^2\\
\mathrm{s.t.}\quad&\cref{eq:constrfinfty},\ \cref{eq:constrener},\ \mathrm{and}\ \cref{eq:constrNa}\ \mathrm{hold},
\end{aligned}
\right\}
\end{align}
where the additional constraint is
\begin{align}\label{eq:constrNa}
\frac{1}{2}\suma\Na\left(\f\right)^2\leq\mathcal L\left(u\right)\coloneqq\frac{8\pi}{\sigma}\suma\left|\e\right|^2\left\|\mathring\f\right\|_{L^\infty\left(\left[0,T\right]\times\Omega\times\R^3\right)}^2\mathcal I\left(u\right);
\end{align}
see \cref{eq:constrener} for the definition of $\mathcal I\left(u\right)$. On the one hand, \cref{eq:constrNa} is automatically satisfied if $\mathcal G\left(y,u\right)=0$ and \cref{eq:constrfinfty,eq:constrener} hold due to \cref{eq:estNa}. Hence, feasible points for \cref{prob:P} are also feasible for \cref{prob:Ps}. On the other hand, \cref{eq:constrNa} ensures a certain weak lower semi-continuity of $\left\|\mathcal G\right\|_{\Lambda^*}$ by the following lemma (and this is conversely the very reason why we impose \cref{eq:constrNa}):
\begin{lemma}\label{lma:Gwlsc}
	Let $\left(\left(y_k,u_k\right)\right)\subset\mathcal Y\times\mathcal U$ with $\f_k\geq 0$ and $u\in\mathcal U$, $\f\in L^\infty\left(\left[0,T\right]\times\Omega\times\R^3\right)$, $\f_+\in L^q\left(\gamma_{T}^+,d\gamm\right)$, $\left(E,H\right)\in L^2\left(\left[0,T\right]\times\Omega;\R^6\right)$ such that for $k\to\infty$ it holds that: $u_k\wto u$ in $\mathcal U$, $\f_k\overset{*}\wto\f$ in $L^\infty\left(\left[0,T\right]\times\Omega\times\R^3\right)$, $\f_{k,+}\wto\f_+$ in $L^q\left(\gamma_{T}^+,d\gamm\right)$, $\left(E_k,H_k\right)\wto\left(E,H\right)$ in $L^2\left(\left[0,T\right]\times\Omega;\R^6\right)$. Furthermore, assume that \cref{eq:constrener,eq:constrNa} are satisfied along the sequence. Then $\left(y,u\right)\in\mathcal Y\times\mathcal U$, \cref{eq:constrener,eq:constrNa} are preserved in the limit, and there holds
	\begin{align}\label{eq:Gwlsc}
	\left\|\mathcal G\left(y,u\right)\right\|_{\Lambda^*}\leq\liminf_{k\to\infty}\left\|\mathcal G\left(y_k,u_k\right)\right\|_{\Lambda^*}.
	\end{align}
\end{lemma}
\begin{proof}
	Note that $\left(u_k\right)$ converges to $u$ strongly in $L^2\left(\left[0,T\right]\times\Gamma;\R^3\right)$.
	
	\textit{Step 1:} $\f\in\Y$: Consider $g_k\coloneqq\partial_t\f_k+\v\cdot\partial_x\f_k$. In light of \cref{eq:constrNa} and the boundedness of $\left(u_k\right)$, the sequence $\left(g_k\right)$ is bounded in $L^2\left(\left[0,T\right]\times\Omega;H^{-1}\left(\R^3\right)\right)$. Therefore, $\left(g_k\right)$ converges, after possibly extracting a suitable subsequence, to some $g$ weakly-* in $L^2\left(\left[0,T\right]\times\Omega;H^{-1}\left(\R^3\right)\right)$. Since for all $\varphi\in C_c^\infty\left(\left]0,T\right[\times\Omega\times\R^3\right)$
	\begin{align*}
	g\left(\varphi\right)&=\lim_{k\to\infty}\left(\partial_t\f_k+\v\cdot\partial_x\f_k\right)\left(\varphi\right)=\lim_{k\to\infty}-\int_0^{T}\int_\Omega\int_{\R^3}\left(\f_k\partial_t\varphi+\v\f_k\cdot\partial_x\varphi\right)\,dvdxdt\\
	&=-\int_0^{T}\int_\Omega\int_{\R^3}\left(\f\partial_t\varphi+\v\f\cdot\partial_x\varphi\right)\,dvdxdt=\left(\partial_t\f+\v\cdot\partial_x\f\right)\left(\varphi\right)
	\end{align*}
	and since $C_c^\infty\left(\left]0,T\right[\times\Omega\times\R^3\right)$ is dense in $L^2\left(\left[0,T\right]\times\Omega;H^1\left(\R^3\right)\right)$, we have 
	\begin{align*}
	\partial_t\f+\v\cdot\partial_x\f=g\in L^2\left(\left[0,T\right]\times\Omega;H^{-1}\left(\R^3\right)\right).
	\end{align*}
	As in the proof of \cref{thm:ExistenceMinimizer}, we also see that $\f\in\left(L_\kin^1\cap L^\infty\right)\left(\left[0,T\right]\times\Omega\times\R^3\right)$ and that \cref{eq:constrener} is preserved in the limit. Altogether, $\left(y,u\right)\in\mathcal Y\times\mathcal U$.
	
	\textit{Step 2:} \cref{eq:constrNa} is preserved in the limit as well: Let $\iota>0$. By $u_k\to u$ in $L^2\left(\left[0,T\right]\times\Gamma;\R^3\right)$ we have $\mathcal I\left(u_k\right)\to\mathcal I\left(u\right)$ for $k\to\infty$, whence
	\begin{align*}
	\frac{1}{2}\suma\Na\left(\f_k\right)^2\leq\mathcal L\left(u\right)+\iota
	\end{align*}
	for large $k$. By the weak-*-convergence obtained in Step 1, there holds
	\begin{align*}
	\Na\left(\f\right)\leq\liminf_{k\to\infty}\Na\left(\f_k\right).
	\end{align*}
	Thus,
	\begin{align*}
	\frac{1}{2}\suma\Na\left(\f\right)^2\leq\liminf_{k\to\infty}\frac{1}{2}\suma\Na\left(\f_k\right)^2\leq\mathcal L\left(u\right)+\iota.
	\end{align*}
	Since $\iota>0$ was arbitrary, we are done.
	
	\textit{Step 3:} Proof of \cref{eq:Gwlsc}: To this end, we have to pass to the limit in the right-hand sides of \cref{eq:Vlasovweak,eq:Maxwellweak}; this procedure has already been carried out in a similar situation, see the proof of \cref{thm:ExistenceMinimizer}. As a consequence of \cref{lma:jintconstr}, we may assume that $\left(j_k^\inte\right)$ converges weakly to $j^\inte$ in $L^{\frac{4}{3}}\left(\left[0,T\right]\times\Omega;\R^3\right)$; in order to verify that this weak limit indeed is $j^\inte$, we recall that an energy estimate like \cref{eq:constrener} is sufficient, see proof of \cref{thm:ExistenceMinimizer}. Hence, we can easily pass to the limit in all terms but the nonlinear one, first for $\left(\left(\psia\right)_\alpha,\vartheta^e,\vartheta^h\right)\in\Psi_{T}^N\times\Theta_{T}^2$ and then for arbitrary $\left(\left(\psia\right)_\alpha,\vartheta^e,\vartheta^h\right)\in\Lambda$ with the help of \cref{sec:somestimates}. Regarding the nonlinear term, we first consider $\psia\in\Psi_{T}$ that factorizes, i.e., $\psia\left(t,x,v\right)=\psia_1\left(t,x\right)\psia_2\left(v\right)$. For some $\iota>0$ and $\zeta\in C_c^\infty\left(\R^3\right)$ with $\supp\zeta\subset B_r$ (for some $r>0$), we find an $\eta\in C_c^\infty\left(\left]0,T\right[\times\Omega\times B_r\right)$ such that
	\begin{align}\label{eq:etaestwsc}
	\left\|\int_{\R^3}\zeta\left(v\right)\left(\left(1-\eta\right)\left(\f_k-\f\right)\right)\left(\cdot,\cdot,v\right)\,dv\right\|_{L^2\left(\left[0,T\right]\times\Omega\right)}<\iota;
	\end{align}
	note that the $L^2$-norms of the $\f_k$ are uniformly bounded. For this fixed $\eta$, there holds
	\begin{align*}
	\partial_t\left(\eta\f_k\right)+\v\cdot\partial_x\left(\eta\f_k\right)=\f_k\partial_t\eta+\f_k\v\cdot\partial_x\eta+\eta g_k-h_k\cdot\partial_v\eta+\div_v\left(\eta h_k\right)
	\end{align*}
	on whole $\R\times\R^3\times\R^3$ in the sense of distributions, where $g_k\in L^2\left(\left[0,T\right]\times\Omega\times\R^3\right)$, $h_k\in L^2\left(\left[0,T\right]\times\Omega\times\R^3;\R^3\right)$ are chosen such that
	\begin{align*}
	\partial_t\f_k+\v\cdot\partial_x\f_k&=g_k+\div_vh_k,\\
	\left\|\partial_t\f_k+\v\cdot\partial_x\f_k\right\|_{L^2\left(\left[0,T\right]\times\Omega;H^{-1}\left(\R^3\right)\right)}&=\left(\left\|g_k\right\|_{L^2\left(\left[0,T\right]\times\Omega\times\R^3\right)}^2+\left\|h_k\right\|_{L^2\left(\left[0,T\right]\times\Omega\times\R^3;\R^3\right)}^2\right)^{\frac{1}{2}}.
	\end{align*} 
	Again by the uniform boundedness of the $\f_k$ in $L^2\left(\left[0,T\right]\times\Omega\times\R^3\right)$, we can easily estimate
	\begin{align*}
	\left\|\partial_t\left(\eta\f_k\right)+\v\cdot\partial_x\left(\eta\f_k\right)\right\|_{L^2\left(\R\times\R^3;H^{-1}\left(\R^3\right)\right)}\leq C\left(\eta\right)\left(1+\Na\left(\f_k\right)\right).
	\end{align*}
	By virtue of \cref{eq:constrNa} and the boundedness of $\left(u_k\right)$, the right-hand side is uniformly bounded in $k$, whence we have for a subsequence,
	\begin{align*}
	\int_{\R^3}\zeta\left(v\right)\left(\eta\f_{k_j}\right)\left(\cdot,\cdot,v\right)\,dv\overset{j\to\infty}\longrightarrow\int_{\R^3}\zeta\left(v\right)\left(\eta\f\right)\left(\cdot,\cdot,v\right)\,dv
	\end{align*}
	in $L^2\left(\left[0,T\right]\times\Omega\right)$ due to \cref{lma:MomentumAveraging}. Assuming that all $\psia\in\Psi_{T}$ factorize and using \cref{eq:etaestwsc}, we may now pass to the limit in all terms along a common subsequence, that is
	\begin{align*}
	\mathcal G\left(y,u\right)\left(\left(\psia\right)_\alpha,\vartheta^e,\vartheta^h\right)=\lim_{j\to\infty}\mathcal G\left(y_{k_j},u_{k_j}\right)\left(\left(\psia\right)_\alpha,\vartheta^e,\vartheta^h\right).
	\end{align*}
	Since the limit on the left-hand side does not depend on the extraction of this subsequence, we conclude that the equality above even holds for the full limit $k\to\infty$ by using the standard subsubsequence argument. Thus,
	\begin{align*}
	\left|\mathcal G\left(y,u\right)\left(\left(\psia\right)_\alpha,\vartheta^e,\vartheta^h\right)\right|\leq\liminf_{k\to\infty}\left\|\mathcal G\left(y_k,u_k\right)\right\|_{\Lambda^*}\left\|\left(\left(\psia\right)_\alpha,\vartheta^e,\vartheta^h\right)\right\|_\Lambda.
	\end{align*}
	This inequality then also holds for general $\left(\left(\psia\right)_\alpha,\vartheta^e,\vartheta^h\right)\in\Lambda$ by a density argument (cf. proof of \cref{thm:ExistenceMinimizer} and the definition of $\Lambda$). Altogether, \cref{eq:Gwlsc} is proved.
\end{proof} 
\begin{remark}
	It is important to understand the necessity of \cref{eq:constrNa} for \cref{lma:Gwlsc} and for later treating \cref{prob:Ps}: In \cite{Web19}, \cref{lma:MomentumAveraging} was applied to a sequence where any $\f_k$ \textit{already solves} a Vlasov equation in the sense of distributions, that is
	\begin{align*}
	\partial_t\f_k+\v\cdot\partial_x\f_k=-\div_v\left(F_k\f_k\right),
	\end{align*}
	which gave a direct estimate on the $L^2\left(\left[0,T\right]\times\Omega;H^{-1}\left(\R^3\right)\right)$-norm of $\f_k$ by the corresponding a priori $L^p$-bounds on $F_k$ and $\f_k$. However, the $\f$ of some $\left(y,u\right)$ that is feasible for \cref{prob:Ps} do not necessarily solve a Vlasov equation as above. Thus, suitable estimates on the $L^2\left(\left[0,T\right]\times\Omega;H^{-1}\left(\R^3\right)\right)$-norm along some sequence can not be obtained without imposing them a priori, that is, imposing \cref{eq:constrNa}. Without this, we would not be able to pass to the limit as in the proof above, and the important weak lower semi-continuity of $\left\|\mathcal G\right\|_{\Lambda^*}$ could not be proved.
\end{remark}
Now we are able to prove existence of minimizers of \cref{prob:Ps}:
\begin{theorem}\label{thm:ExistenceMinimizerPs}
	There is a (not necessarily unique) minimizer of \cref{prob:Ps}.
\end{theorem}
\begin{proof}
	This is proved in much the same way as \cref{thm:ExistenceMinimizer} was proved. We no longer have to show that \eqref{eq:WholeSystem} has to be preserved in the limit. Instead, we apply \cref{lma:Gwlsc}: The assumptions there are satisfied for a minimizing sequence (after extracting a suitable subsequence) and the respective weak limits. Thus, the new constraint \cref{eq:constrNa} is also preserved in the limit and the new objective function $\mathcal J_s$ indeed admits its minimum at the limit tuple $\left(y,u\right)$.
\end{proof}
Later, we will need that $\mathcal Y\times\mathcal U$ is complete; this is proved in the following lemma:
\begin{lemma}\label{lma:YUBanach}
	$\mathcal Y\times\mathcal U$ is a Banach space.
\end{lemma}
\begin{proof}
	We only have to show completeness of $\Y$: Let $\left(f_k\right)$ be a Cauchy sequence in $\Y$. Clearly, this sequence converges to some $f$ with respect to the $L^1$- and $L^\infty$-norm. Since $L^2\left(\left[0,T\right]\times\Omega;H^{-1}\left(\R^3\right)\right)$ is complete, the sequence $\left(\partial_tf_k+\v\cdot\partial_xf_k\right)$ converges to some $g$ in this space. As in Step 1 of the proof of \cref{lma:Gwlsc}, we see that $g=\partial_tf+\v\cdot\partial_xf$. Thus, $\left(f_k\right)$ converges to $f$ in the whole $\Y$-norm.
\end{proof}

Next, we want to derive first order optimality conditions for a minimizer of \cref{prob:Ps}. To this end, we consider the differentiability of the objective function $\mathcal J_s$. Clearly, the only difficult term is $\left\|\mathcal G\left(y,u\right)\right\|_{\Lambda^*}^2$. To tackle this one, we state a duality result, which links differentiability of a norm to uniform convexity of the dual space:
\begin{proposition}\label{prop:smoothconvex}
	A Banach space $X$ is uniformly smooth if and only if $X^*$ is uniformly convex. In this case, for each unit vector $x\in X$ there is exactly one $x^*\in X^*$ with $\left\|x^*\right\|_{X^*}=1$ satisfying $x^*x=1$. Furthermore, this $x^*$ is the derivative of the norm at $x$.
\end{proposition}
Here, \enquote{uniformly smooth} means that
\begin{align*}
\lim_{t\to 0}\frac{\left\|x+ty\right\|_X-\left\|x\right\|_X}{t}
\end{align*}
exists and is uniform in $x,y\in\left\{z\in X\mid\left\|z\right\|_X=1\right\}$. The original work in this subject was done by Day \cite{Day44}; see also \cite[Chapter 2]{Lin03} for an overview of different concepts of and relations between convexity and smoothness of normed spaces.

From \cref{prop:smoothconvex} we easily get the following corollary, which we will need in the following:
\begin{corollary}\label{cor:diffnormsq}
	Let $X$ be a Banach space such that $X^*$ is uniformly convex. Then the map $z\colon X\to\R$, $z\left(x\right)=\frac{1}{2}\left\|x\right\|_X^2$ is differentiable on whole $X$ with derivative $z'\left(x\right)=x^*$ where $x^*$ is the unique element of $X^*$ satisfying $\left\|x^*\right\|_{X^*}=\left\|x\right\|_X$ and $x^*x=\left\|x\right\|_X^2$. (The map $z'\colon X\to X^*$ is often referred to as the duality map.)
\end{corollary}
\begin{proof}
	By \cref{prop:smoothconvex}, the norm is differentiable on the unit sphere of $X$. Since the norm is positive homogeneous, this holds true on whole $X$ except in $x=0$, and the derivative is $x^*$ such that $\left\|x^*\right\|_{X^*}=1$ and $x^*x=\left\|x\right\|_X$ (still this $x^*$ is uniquely determined by these two properties). Applying the chain rule we see that $z$ is differentiable on $X\setminus\left\{0\right\}$ and has the asserted derivative.
	
	That $z$ is differentiable in $x=0$ and $z'\left(0\right)=0$ is clear.
\end{proof}
With this corollary we see that the objective function $\mathcal J_s$ is differentiable:
\begin{lemma}\label{lma:Jsdiff}
	The objective function $\mathcal J_s$ is differentiable and its derivative is given by
	\begin{align}\label{eq:Jsdiff}
	&\mathcal J_s'\left(y,u\right)\left(\delta y,\delta u\right)\nonumber\\
	&=\suma\w\int_{\gamma_{T}^+}\sign{\f_+}\left|\f_+\right|^{q-1}\delta\f_+\,d\gamm\nonumber\\
	&\phantom{=\;}+\sum_{j=1}^3\int_0^{T}\int_\Gamma\left(\vphantom{\sum_{i=1}^3}\sign{u_j}\left|u_j\right|^{r-1}\delta u_j+\kappa_1\sign{\partial_t u_j}\left|\partial_tu_j\right|^{r-1}\partial_t\delta u_j\right.\nonumber\\
	&\omit\hfill$\displaystyle\left.+\kappa_2\sum_{i=1}^3\sign{\partial_{x_i}u_j}\left|\partial_{x_i}u_j\right|^{r-1}\partial_{x_i}\delta u_j\right)\,dxdt$\nonumber\\
	&\phantom{=\;}+\suma\left(-\int_0^{T}\int_\Omega\int_{\R^3}\left(\left(\partial_t\psia+\v\cdot\partial_x\psia+\e\left(E+\v\times H\right)\cdot\partial_v\psia\right)\delta\f\right.\right.\nonumber\\
	&\omit\hfill$\displaystyle\left.+\e\left(\delta E+\v\times\delta H\right)\f\cdot\partial_v\psia\right)\,dvdxdt$\nonumber\\
	&\phantom{=\;}\phantom{=\suma\Big(}\left.+\int_{\gamma_{T}^+}\delta\f_+\psia\,d\gamm-\int_{\gamma_{T}^-}\a\left(K\delta\f_+\right)\psia\,d\gamm\right)\nonumber\\
	&\phantom{=\;}+\int_0^{T}\int_{\R^3}\left(\varepsilon\delta E\cdot\partial_t\vartheta^e-\delta H\cdot\curl_x\vartheta^e-4\pi\left(\delta j^\inte+\delta u\right)\cdot\vartheta^e\right)\,dxdt\nonumber\\
	&\phantom{=\;}+\int_0^{T}\int_{\R^3}\left(\mu\delta H\cdot\partial_t\vartheta^h+\delta E\cdot\curl_x\vartheta^h\right)\,dxdt,
	\end{align}
	where $\left(\left(\psia\right)_\alpha,\vartheta^e,\vartheta^h\right)\in\Lambda$ is the unique element in $\Lambda$ satisfying
	\begin{subequations}\label{eq:Jsdiffdual}
		\begin{align}
		\left\|\left(\left(\psia\right)_\alpha,\vartheta^e,\vartheta^h\right)\right\|_{\Lambda}&=s\left\|\mathcal G\left(y,u\right)\right\|_{\Lambda^*},\\
		\mathcal G\left(y,u\right)\left(\left(\psia\right)_\alpha,\vartheta^e,\vartheta^h\right)&=s\left\|\mathcal G\left(y,u\right)\right\|_{\Lambda^*}^2.
		\end{align}
	\end{subequations}
\end{lemma}
\begin{proof}
	The only difficult term is $\frac{s}{2}\left\|\mathcal G\left(y,u\right)\right\|_{\Lambda^*}^2$. The other terms are easy to handle in a standard way.
	
	Denoting $Z\left(y,u\right)=\frac{s}{2}\left\|\mathcal G\left(y,u\right)\right\|_{\Lambda^*}^2$ we apply \cref{lma:Gdiff,cor:diffnormsq}. The latter is applicable since the dual of $\Lambda^*$, that is $\Lambda^{**}\cong\Lambda$, is uniformly convex due to \cref{lma:Lambdauniconrefl}. At this point we should mention that this step is exactly the reason why we work with a uniformly convex, reflexive test function space. Hence, additionally using the chain rule, we see that $Z$ is differentiable with
	\begin{align}\label{eq:Zdiff}
	Z'\left(y,u\right)\left(\delta y,\delta u\right)=s\lambda^{**}\mathcal G'\left(y,u\right)\left(\delta y,\delta u\right),
	\end{align}
	where $\lambda^{**}\in\Lambda^{**}$ uniquely satisfies
	\begin{subequations}\label{eq:lambdadual}
		\begin{align}
		\left\|\lambda^{**}\right\|_{\Lambda^{**}}&=\left\|\mathcal G\left(y,u\right)\right\|_{\Lambda^*},\\
		\lambda^{**}\mathcal G\left(y,u\right)&=\left\|\mathcal G\left(y,u\right)\right\|_{\Lambda^*}^2.
		\end{align}
	\end{subequations}
	Since $\Lambda$ is reflexive, we can regard $\lambda^{**}$ as a $\lambda\in\Lambda$ via the canonical isomorphism. We define $\left(\left(\psia\right)_\alpha,\vartheta^e,\vartheta^h\right)$ by multiplying this $\lambda$ with the positive number $s$. On the one hand, from \cref{eq:Zdiff} we get the remaining part of \cref{eq:Jsdiff}, that is 
	\begin{align*}
	\mathcal G'\left(y,u\right)\left(\delta y,\delta u\right)\left(\left(\psia\right)_\alpha,\vartheta^e,\vartheta^h\right),
	\end{align*}
	which is given by \cref{eq:Gdiff}. On the other hand, \cref{eq:lambdadual} instantly yields \cref{eq:Jsdiffdual}.
\end{proof}
\begin{remark}
	Such a $\left(\left(\psia\right)_\alpha,\vartheta^e,\vartheta^h\right)$ will later act as a Lagrangian multiplier with respect to the Vlasov-Maxwell system, that is, a solution of the adjoint equation, if the point $\left(y,u\right)$ is a minimizer of \cref{prob:Ps}.
\end{remark}
Next, we derive necessary first order optimality conditions for \cref{prob:Ps}. To tackle an optimization problem with certain constraints and to prove existence of Lagrangian multipliers with respect to them, one has to verify some constraint qualification. To this end, we state a famous result of Zowe and Kurcyusz \cite{ZK79}, which is based on a fundamental work of Robinson \cite{Rob76}:
\begin{proposition}\label{prop:zowekur}
	Let $X$, $Y$ be Banach spaces, $C\subset X$ non-empty, closed, and convex, $K\subset Y$ a closed convex cone, $\phi\colon X\rightarrow\mathbb{R}$ differentiable, and $g\colon X\rightarrow Y$ continuously differentiable. Denote for $A\subset X$ (and similarly for $A\subset Y$)
	\begin{align*}
	A^+=\left\{x^*\in X^*\mid x^*a\geq 0\ \forall a\in A\right\}
	\end{align*}
	and denote for $x\in X$ and $y\in Y$
	\begin{align*}
	C\left(x\right)&=\left\{\lambda\left(c-x\right)\mid c\in C,\lambda\geq 0\right\},\\
	K\left(y\right)&=\left\{k-\lambda y\mid k\in K,\lambda\geq 0\right\}.
	\end{align*}
	Let $x_*\in X$ be a local minimizer (i.e., a local minimizer of the objective function restricted to all feasible points) of the problem
	\begin{align*}
	\min_{x\in X}\quad\phi\left(x\right)\qquad\mathrm{s.t.}\quad x\in C,g\left(x\right)\in K,
	\end{align*}
	and let the constraint qualification
	\begin{align}\tag{CQ}\label{eq:CQ}
	g'\left(x_*\right)C\left(x_*\right)-K\left(g\left(x_*\right)\right)=Y
	\end{align}
	hold.
	
	Then there is a Lagrange multiplier $y^*\in Y^*$ at $x_*$ for the problem above, i.e.,
	\begin{enumerate}[label=(\roman*),leftmargin=*]
		\item\label{prop:zowekuri} $y^*\in K^+$,
		\item\label{prop:zowekurii} $y^*g\left(x_*\right)=0$,
		\item\label{prop:zowekuriii} $\phi'\left(x_*\right)-y^*\circ g'\left(x_*\right)\in C\left(x_*\right)^+$.
	\end{enumerate}
\end{proposition}
We apply this result to our problem \cref{prob:Ps}. As we have shown in \cref{lma:Jsdiff}, the objective function is differentiable. In the following, let
\begin{align*}
C&\coloneqq\left\{\left(y,u\right)\in\mathcal Y\times\mathcal U\mid 0\leq\f\leq\left\|\mathring\f\right\|_{L^\infty\left(\Omega\times\R^3\right)}\ \mathrm{a.e.},\alpha=1,\dots,N\right\}\subset\mathcal Y\times\mathcal U,\\
K&\coloneqq\R_{\geq 0}^2\subset\R^2.
\end{align*}
Clearly, $C$ is non-empty, closed, and convex, and $K$ is a closed convex cone. Furthermore, the constraints \cref{eq:constrfinfty,eq:constrener,eq:constrNa} are equivalent to
\begin{align*}
\left(y,u\right)\in C, g\left(y,u\right)\in K,
\end{align*}
where
\begin{align*}
g_1\left(y,u\right)&=\mathcal I\left(u\right)-\suma\int_0^{T}\int_\Omega\int_{\R^3}\vo\f\,dvdxdt-\frac{\sigma}{8\pi}\int_0^{T}\int_{\R^3}\left(\left|E\right|^2+\left|H\right|^2\right)\,dxdt,\\
g_2\left(y,u\right)&=\mathcal L\left(u\right)-\frac{1}{2}\suma\Na\left(\f\right)^2.
\end{align*}
It is easy to see that $g$ is continuously differentiable with
\begin{align*}
g_1'\left(y,u\right)\left(\delta y,\delta u\right)&=\beta_1\int_0^{T}\int_\Gamma u\cdot\delta u\,dxdt-\suma\int_0^{T}\int_\Omega\int_{\R^3}\vo\delta\f\,dvdxdt\\
&\phantom{=\;}-\frac{\sigma}{4\pi}\int_0^{T}\int_{\R^3}\left(E\cdot\delta E+H\cdot\delta H\right)\,dxdt,\\
g_2'\left(y,u\right)\left(\delta y,\delta u\right)&=\beta_2\int_0^{T}\int_\Gamma u\cdot\delta u\,dxdt\\
&\phantom{=\;}-\suma\int_0^{T}\int_\Omega\left(\partial_t\delta\f+\v\cdot\partial_x\delta\f\right)\left(\mathcal R^{-1}\left(\partial_t\f+\v\cdot\partial_x\f\right)\right)\,dxdt,
\end{align*}
where $\mathcal R\colon H^1\left(\R^3\right)\to H^{-1}\left(\R^3\right)$ is the Riesz isomorphism and 
\begin{align*}
\mathcal R^{-1}\left(\partial_t\f+\v\cdot\partial_x\f\right)\equiv\mathcal R^{-1}\left(\left(\partial_t\f+\v\cdot\partial_x\f\right)\left(t,x\right)\right),
\end{align*}
and where
\begin{align*}
\beta_1\coloneqq 4\pi T^2\sigma^{-1},\quad
\beta_2\coloneqq\beta_1\frac{8\pi}{\sigma}\suma\left|\e\right|^2\left\|\mathring\f\right\|_{L^\infty\left(\left[0,T\right]\times\Omega\times\R^3\right)}^2.
\end{align*}
We verify the constraint qualification \cref{eq:CQ}:
\begin{lemma}\label{lma:CQ}
	Let $\left(y_s,u_s\right)$ be a minimizer of \cref{prob:Ps}. Then \cref{eq:CQ} is satisfied if $s$ is sufficiently large.
\end{lemma}
\begin{proof}
	First, we exclude the possibility that some $\f_s$ is identically zero for $s$ sufficiently large (since then the term $\frac{s}{2}\left\|\mathcal G\left(y_s,u_s\right)\right\|_{\Lambda^*}^2$ is too large for $\left(y_s,u_s\right)$ to be a minimizer of \cref{prob:Ps}): For each $\alpha$, 
	let $\psia_*\colon\left[0,T\right]\times\Omega\times\R^3\to\R$, $\psia_*\left(t,x,v\right)=\eta\left(t\right)\varphi^\alpha\left(x,v\right)$, where $\eta\in C^\infty\left(\left[0,T\right]\right)$ with $\eta\left(0\right)=1$ and $\supp\eta\subset\left[0,T\right[$, and $\varphi^\alpha\in C_c^\infty\left(\Omega\times\R^3\right)$ with
	\begin{align*}
	\left\|\mathring\f-\varphi^\alpha\right\|_{L^2\left(\Omega\times\R^3\right)}\leq\frac{1}{2}\left\|\mathring\f\right\|_{L^2\left(\Omega\times\R^3\right)}.
	\end{align*}
	Clearly, $\psia_*\in\Psi_{T}$. Now assume $f_s^{\alpha_0}=0$ for some $\alpha_0$. There holds
	\begin{align*}
	&\left|\mathcal G\left(y_s,u_s\right)\left(\left(0,\dots,0,\psi_*^{\alpha_0},0,\dots,0\right),0,0\right)\right|=\left|\int_\Omega\int_{\R^3}\mathring f^{\alpha_0}\varphi^{\alpha_0}\,dvdx\right|\\
	&=\left|\left\|\mathring f^{\alpha_0}\right\|_{L^2\left(\Omega\times\R^3\right)}^2-\int_\Omega\int_{\R^3}\mathring f^{\alpha_0}\left(\mathring f^{\alpha_0}-\varphi^{\alpha_0}\right)\,dvdx\right|\geq\frac{1}{2}\left\|\mathring f^{\alpha_0}\right\|_{L^2\left(\Omega\times\R^3\right)}^2.
	\end{align*}
	Thus,
	\begin{align*}
	\left\|\mathcal G\left(y_s,u_s\right)\right\|_{\Lambda^*}\geq\frac{\left\|\mathring f^{\alpha_0}\right\|_{L^2\left(\Omega\times\R^3\right)}^2}{2\left\|\psi_*^{\alpha_0}\right\|_{W^{1,p,p'}}},
	\end{align*}
	whence
	\begin{align}\label{eq:Jscontr}
	\mathcal J_s\left(y_s,u_s\right)\geq s\cdot\frac{\left\|\mathring f^{\alpha_0}\right\|_{L^2\left(\Omega\times\R^3\right)}^4}{8\left\|\psi_*^{\alpha_0}\right\|_{W^{1,p,p'}}^2}>\mathcal J\left(y_*,u_*\right)=\mathcal J_s\left(y_*,u_*\right),
	\end{align}
	where $\left(y_*,u_*\right)$ is a minimizer of \cref{prob:P} and where the strict inequality holds for $s$ sufficiently large, i.e.,
	\begin{align*}
	s>\max_{\alpha=1,\dots,N}\frac{8\left\|\psia_*\right\|_{W^{1,p,p'}}^2\mathcal J\left(y_*,u_*\right)}{\left\|\mathring\f\right\|_{L^2\left(\Omega\times\R^3\right)}^4};
	\end{align*}
	note that the right-hand side does not depend on $s$ and $\alpha_0$ and that no $\mathring\f$ is identically zero. Since $\left(y_*,u_*\right)$ is feasible for \cref{prob:Ps}, \cref{eq:Jscontr} is a contradiction to $\left(y_s,u_s\right)$ being a minimizer of \cref{prob:Ps}.
	
	To prove the lemma, we have to show that for each $d\in\R^2$ there are $\lambda_1,\lambda_2\geq 0$, $k\in\R_{\geq 0}^2$, and $\left(\delta y,\delta u\right)\in C$ satisfying
	\begin{align}\label{eq:CQappl}
	\lambda_1g'\left(y_s,u_s\right)\left(\delta y-y_s,\delta u-u_s\right)-k+\lambda_2g\left(y_s,u_s\right)=d.
	\end{align}
	We choose $\delta\f_+=\f_{s,+}$ for all $\alpha$, $\delta E=E_s$, $\delta H=H_s$, $\delta u=u_s$, and consider seven cases; note that in the following there always holds $\lambda_1,\lambda_2\geq 0$, $k\in\R_{\geq 0}^2$, and $\left(\delta y,\delta u\right)\in C$:\\
	\textit{Case 1:} $d_1,d_2\leq 0$: Choose $\lambda_1=\lambda_2=0$, $\delta\f=\f_s$ for all $\alpha$, $k=-d$.\\
	\textit{Case 2:} $d_1>0$, $d_2\leq 0$: Choose $\lambda_2=0$, $\delta f^1=\frac{1}{2}f_s^1$, $\delta\f=\f_s$ for $\alpha\geq 2$, $k_1=0$. Since
	\begin{align*}
	g_1'\left(y_s,u_s\right)\left(\delta y-y_s,\delta u-u_s\right)=\frac{1}{2}\int_0^{T}\int_\Omega\int_{\R^3}v_1^0f_s^1\,dvdxdt>0,
	\end{align*}
	we can choose $\lambda_1>0$ such that the first component of \cref{eq:CQappl} is satisfied. Then set
	\begin{align*}
	k_2=-d_2+\lambda_1g_2'\left(y_s,u_s\right)\left(\delta y-y_s,\delta u-u_s\right)=-d_2+\frac{\lambda_1}{2}\mathcal N_1\left(f_s^1\right)^2.
	\end{align*}
	\textit{Case 3:} $d_1\leq 0$, $d_2>0$, $g_2\left(y_s,u_s\right)=0$: Because of $\mathcal L\left(u_s\right)>0$, there holds $\mathcal N_{\alpha_0}\left(f_s^{\alpha_0}\right)>0$ for some $\alpha_0$. Choose $\lambda_2=0$, $\delta f^{\alpha_0}=\frac{1}{2}f_s^{\alpha_0}$, $\delta\f=\f_s$ for $\alpha\neq\alpha_0$, $k_2=0$. Since
	\begin{align*}
	g_2'\left(y_s,u_s\right)\left(\delta y-y_s,\delta u-u_s\right)=\frac{1}{2}\mathcal N_{\alpha_0}\left(f_s^{\alpha_0}\right)^2>0,
	\end{align*}
	we can choose $\lambda_1>0$ such that the second component of \cref{eq:CQappl} is satisfied. Then set
	\begin{align*}
	k_1=-d_1+\lambda_1g_1'\left(y_s,u_s\right)\left(\delta y-y_s,\delta u-u_s\right)=-d_1+\frac{\lambda_1}{2}\int_0^{T}\int_\Omega\int_{\R^3}v_{\alpha_0}^0f_s^{\alpha_0}\,dvdxdt.
	\end{align*}
	\textit{Case 4:} $d_1\leq 0$, $d_2>0$, $g_2\left(y_s,u_s\right)>0$: Choose $\lambda_1=0$, $\delta\f=\f_s$ for all $\alpha$, $k_2=0$. By $g_2\left(y_s,u_s\right)>0$ we can choose $\lambda_2>0$ such that the second component of \cref{eq:CQappl} is satisfied. Then set $k_1=-d_1+\lambda_2g_1\left(y_s,u_s\right)$.\\
	\textit{Case 5:} $d_1,d_2>0$, $g_2\left(y_s,u_s\right)=0$: As in Case 3, there holds $\mathcal N_{\alpha_0}\left(f_s^{\alpha_0}\right)>0$ for some $\alpha_0$. Choose $\lambda_2=0$, $\delta f^{\alpha_0}=\frac{1}{2}f_s^{\alpha_0}$, $\delta\f=\f_s$ for $\alpha\neq\alpha_0$, $k_2=0$. Since
	\begin{align*}
	g_1'\left(y_s,u_s\right)\left(\delta y-y_s,\delta u-u_s\right)&=\frac{1}{2}\int_0^{T}\int_\Omega\int_{\R^3}v_{\alpha_0}^0f_s^{\alpha_0}\,dvdxdt>0,\\
	g_2'\left(y_s,u_s\right)\left(\delta y-y_s,\delta u-u_s\right)&=\frac{1}{2}\mathcal N_{\alpha_0}\left(f_s^{\alpha_0}\right)^2>0,
	\end{align*}
	we can choose 
	\begin{align*}
	\lambda_1=\max_{i=1,2}\frac{d_i}{g_i'\left(y_s,u_s\right)\left(\delta y-y_s,\delta u-u_s\right)}
	\end{align*}
	and then $k=-d+\lambda_1g'\left(y_s,u_s\right)\left(\delta y-y_s,\delta u-u_s\right)$.\\
	\textit{Case 6:} $d_1,d_2>0$, $g_2\left(y_s,u_s\right)>0$, $g_1\left(y_s,u_s\right)=0$: Choose $\lambda_2=\frac{d_2}{g_2\left(y_s,u_s\right)}$, $\delta f^1=\frac{1}{2}f_s^1$, $\delta\f=\f_s$ for $\alpha\geq 2$, $k_1=0$. Since
	\begin{align*}
	g_1'\left(y_s,u_s\right)\left(\delta y-y_s,\delta u-u_s\right)=\frac{1}{2}\int_0^{T}\int_\Omega\int_{\R^3}v_1^0f_s^1\,dvdxdt>0,
	\end{align*}
	we can choose $\lambda_1>0$ such that the first component of \cref{eq:CQappl} is satisfied. Then set
	\begin{align*}
	k_2=\lambda_1g_2'\left(y_s,u_s\right)\left(\delta y-y_s,\delta u-u_s\right)=\frac{\lambda_1}{2}\mathcal N_1\left(f_s^1\right)^2.
	\end{align*}
	\textit{Case 7:} $d_1,d_2>0$, $g_2\left(y_s,u_s\right)>0$, $g_1\left(y_s,u_s\right)>0$: Choose $\lambda_1=0$, $\delta\f=\f_s$ for all $\alpha$, and
	\begin{align*}
	\lambda_2=\max_{i=1,2}\frac{d_i}{g_i\left(y_s,u_s\right)},\quad k=-d+\lambda_2g\left(y_s,u_s\right).
	\end{align*}
	In all cases \cref{eq:CQappl} holds; the proof is complete.
\end{proof}
Now, \cref{prop:zowekur} gives us the following theorem:
\begin{theorem}\label{thm:OCs}
	Let $s$ be sufficiently large and $\left(y_s,u_s\right)$ a minimizer of \cref{prob:Ps}. Then there exist two numbers $\nu_s^1,\nu_s^2\geq 0$, and $\taua_s\in\left(\Y\right)^*$, $\alpha=1,\dots,N$, such that:
	\begin{enumerate}[label=(\roman*),leftmargin=*]
		\item\label{lma:optcoJsi} $\nu_s^i=0$ or $g_i\left(y_s,u_s\right)=0$, $i=1,2$;
		\item\label{lma:optcoJsii}
		\begin{align*}
		\suma\taua_s\f_s\leq\suma\taua_s\delta\f
		\end{align*}
		for all $\delta\f\in\Y$ satisfying $0\leq\delta\f\leq\left\|\mathring\f\right\|_{L^\infty\left(\Omega\times\R^3\right)}$ a.e.,
		\item\label{lma:optcoJsiii} for all $\left(\delta y,\delta u\right)\in\mathcal Y\times\mathcal U$ there holds
		\begin{align}\label{eq:optcoJsiii}
		0&=\suma\w\int_{\gamma_{T}^+}\sign{\f_{s,+}}\left|\f_{s,+}\right|^{q-1}\delta\f_+\,d\gamm\nonumber\\
		&\phantom{=\;}+\sum_{j=1}^3\int_0^{T}\int_\Gamma\left(\vphantom{\sum_{i=1}^3}\sign{u_{s,j}}\left|u_{s,j}\right|^{r-1}\delta u_j+\kappa_1\sign{\partial_t u_{s,j}}\left|\partial_tu_{s,j}\right|^{r-1}\partial_t\delta u_j\right.\nonumber\\
		&\omit\hfill$\displaystyle\left.+\kappa_2\sum_{i=1}^3\sign{\partial_{x_i}u_{s,j}}\left|\partial_{x_i}u_{s,j}\right|^{r-1}\partial_{x_i}\delta u_j\right)\,dxdt$\nonumber\\
		&\phantom{=\;}+\suma\left(-\int_0^{T}\int_\Omega\int_{\R^3}\left(\left(\partial_t\psia_s+\v\cdot\partial_x\psia_s+\e\left(E_s+\v\times H_s\right)\cdot\partial_v\psia_s\right)\delta\f\right.\right.\nonumber\\
		&\omit\hfill$\displaystyle\left.+\e\left(\delta E+\v\times\delta H\right)\f_s\cdot\partial_v\psia_s\right)\,dvdxdt$\nonumber\\
		&\phantom{=\;}\phantom{=\suma\Big(}\left.+\int_{\gamma_{T}^+}\delta\f_+\psia_s\,d\gamm-\int_{\gamma_{T}^-}\a\left(K\delta\f_+\right)\psia_s\,d\gamm\right)\nonumber\\
		&\phantom{=\;}+\int_0^{T}\int_{\R^3}\left(\varepsilon\delta E\cdot\partial_t\vartheta_s^e-\delta H\cdot\curl_x\vartheta_s^e-4\pi\left(\delta j^\inte+\delta u\right)\cdot\vartheta_s^e\right)\,dxdt\nonumber\\
		&\phantom{=\;}+\int_0^{T}\int_{\R^3}\left(\mu\delta H\cdot\partial_t\vartheta_s^h+\delta E\cdot\curl_x\vartheta_s^h\right)\,dxdt\nonumber\\
		&\phantom{=\;}+\left(\nu_s^1\beta_1+\nu_s^2\beta_2\right)\int_0^{T}\int_\Gamma u_s\cdot\delta u\,dxdt-\nu_s^1\suma\int_0^{T}\int_\Omega\int_{\R^3}\vo\delta\f\,dvdxdt\nonumber\\
		&\phantom{=\;}-\nu_s^2\suma\int_0^{T}\int_\Omega\left(\partial_t\delta\f+\v\cdot\partial_x\delta\f\right)\left(\mathcal R^{-1}\left(\partial_t\f_s+\v\cdot\partial_x\f_s\right)\right)\,dxdt\nonumber\\
		&\phantom{=\;}-\frac{\nu_s^1\sigma}{4\pi}\int_0^{T}\int_{\R^3}\left(E_s\cdot\delta E+H_s\cdot\delta H\right)\,dxdt-\suma\taua_s\delta\f
		\end{align}
		where $\left(\left(\psia_s\right)_\alpha,\vartheta_s^e,\vartheta_s^h\right)\in\Lambda$ is, in accordance with \cref{eq:Jsdiffdual}, given by
		\begin{align*}
		\left\|\left(\left(\psia_s\right)_\alpha,\vartheta_s^e,\vartheta_s^h\right)\right\|_{\Lambda}&=s\left\|\mathcal G\left(y_s,u_s\right)\right\|_{\Lambda^*},\\
		\mathcal G\left(y_s,u_s\right)\left(\left(\psia_s\right)_\alpha,\vartheta_s^e,\vartheta_s^h\right)&=s\left\|\mathcal G\left(y_s,u_s\right)\right\|_{\Lambda^*}^2.
		\end{align*}
		In other words, \cref{eq:optcoJsiii} can be interpreted as $\left(\left(\psia_s\right)_\alpha,\vartheta_s^e,\vartheta_s^h\right)$ being a solution of the adjoint system
		\begin{subequations}\renewcommand{\theequation}{\mbox{Ad$_{\text s}$}.\arabic{equation}}\phantomsection\makeatletter\def\@currentlabel{\mbox{Ad$_{\text s}$}}\label{eq:Ads}\makeatother
			\begin{align}
			&\partial_t\psia_s+\v\cdot\partial_x\psia_s+\e\left(E_s+\v\times H_s\right)\cdot\partial_v\psia_s\hspace{10mm}\nonumber\\
			&\omit\hfill$\displaystyle=4\pi\v\cdot\vartheta_s^e+\nu_s^1\vo+\nu_s^2\mathcal M_\alpha\f_s+\taua_s$\hspace{10mm}\nonumber\\
			&\omit\hfill$\displaystyle\mathrm{on}\ \left[0,T\right]\times\Omega\times\R^3,$\hspace{10mm}\\
			&\omit$\displaystyle K\a K\psia_{s,-}=\psia_{s,+}+\w\sign{\f_{s,+}}\left|\f_{s,+}\right|^{q-1}\hfill\mathrm{on}\ \gamma_{T}^+,$\hspace{10mm}\\
			&\omit$\displaystyle\psia_s\left(T\right)=0\hfill\mathrm{on}\ \Omega\times\R^3,$\hspace{10mm}\\
			&\omit$\displaystyle\varepsilon\partial_t\vartheta_s^e+\curl_x\vartheta_s^h=-\suma\e\int_{\R^3}\f_s\partial_v\psia_s\,dv-\frac{\nu_s^1\sigma}{4\pi}E_s\hfill\hspace{5mm}\mathrm{on}\ \left[0,T\right]\times\R^3,$\hspace{10mm}\label{eq:Ads4}\\
			&\mu\partial_t\vartheta_s^h-\curl_x\vartheta_s^e=-\suma\e\int_{\R^3}\f_s\left(\partial_v\psia_s\times\v\right)\,dv-\frac{\nu_s^1\sigma}{4\pi}H_s\hspace{10mm}\nonumber\\
			&\omit\hfill$\displaystyle\mathrm{on}\ \left[0,T\right]\times\R^3,$\hspace{10mm}\label{eq:Ads5}\\
			&\omit$\displaystyle\vartheta_s^e\left(T\right)=\vartheta_s^h\left(T\right)=0\hfill\mathrm{on}\ \R^3,$\hspace{10mm}
			\end{align}
		\end{subequations}
		where
		\begin{align*}
		\left(\mathcal M_\alpha\f_s\right)\delta\f=\int_0^{T}\int_\Omega\left(\partial_t\delta\f+\v\cdot\partial_x\delta\f\right)\left(\mathcal R^{-1}\left(\partial_t\f_s+\v\cdot\partial_x\f_s\right)\right)\,dxdt,
		\end{align*}
		and the stationarity condition
		\begin{align}
		0&=\sum_{j=1}^3\int_0^{T}\int_\Gamma\left(\vphantom{\sum_{i=1}^3}\sign{u_{s,j}}\left|u_{s,j}\right|^{r-1}\delta u_j+\kappa_1\sign{\partial_t u_{s,j}}\left|\partial_tu_{s,j}\right|^{r-1}\partial_t\delta u_j\right.\hspace{5mm}\nonumber\\
		&\omit\hfill$\displaystyle\left.+\kappa_2\sum_{i=1}^3\sign{\partial_{x_i}u_{s,j}}\left|\partial_{x_i}u_{s,j}\right|^{r-1}\partial_{x_i}\delta u_j\right)\,dxdt\hspace{5mm}$\nonumber\\
		&\omit$\displaystyle\phantom{=\;}-\int_0^{T}\int_{\Gamma}\left(4\pi\vartheta_s^e-\left(\nu_s^1\beta_1+\nu_s^2\beta_2\right)u_s\right)\cdot\delta u\,dxdt\hfill\mathrm{for\ all}\ \delta u\in\mathcal U\hspace{5mm}$\tag{\mbox{SC$_{\text s}$}}\label{eq:SCs}
		\end{align}
		being satisfied.
	\end{enumerate}
\end{theorem}
\begin{proof}
	Since \cref{eq:CQ} holds due to \cref{lma:CQ} and $\mathcal Y\times\mathcal U$ is a Banach space due to Lemma \ref{lma:YUBanach}, by \cref{prop:zowekur} there is $\nu_s=\left(\nu_s^1,\nu_s^2\right)\in\R^2$ acting as a Lagrangian multiplier with respect to \cref{eq:constrNa}. \Cref{prop:zowekur} \cref{prop:zowekuri} implies $\nu_s^1,\nu_s^2\geq 0$ and \ref{prop:zowekur} \cref{prop:zowekurii} yields \cref{lma:optcoJsi}.
	
	With \cref{prop:zowekur} \cref{prop:zowekuriii} and the notation used there we see that
	\begin{align}\label{eq:taudef}
	\tau_s\coloneqq\mathcal J_s'\left(y_s,u_s\right)-\nu_s\cdot g'\left(y_s,u_s\right)\in C\left(y_s,u_s\right)^+\subset\left(\mathcal Y\times\mathcal U\right)^*.
	\end{align}
	Consequently, $\tau$ can be decomposed into 
	\begin{align*}
	\tau_s&\equiv\left(\left(\taua_s\right)_\alpha,\left(\taua_{s,+}\right)_\alpha,\tau_s^e,\tau_s^h,\tau_s^u\right)\\
	&\in\left(\bigtimes_{\alpha=1}^N\left(\Y\right)^*\times L^q\left(\gamma_{T}^+,d\gamm\right)^*\right)\times\left(L^2\left(\left[0,T\right]\times\R^3;\R^3\right)^*\right)^2\times\mathcal U^*.
	\end{align*}
	Since the cone $C\left(y_s,u_s\right)$ only limits the directions $\delta\f$ and not the directions $\delta\f_+$, $\delta E$, $\delta H$, and $\delta u$, the property $\tau_s\in C\left(y_s,u_s\right)^+$ yields that all $\taua_{s,+}$ and moreover $\tau_s^e$, $\tau_s^h$, and $\tau_s^u$ have to vanish. Thus, $\tau_s\equiv\left(\taua_s\right)_\alpha$ via
	\begin{align}\label{eq:tauident}
	\tau_s\left(\delta y,\delta u\right)=\suma\taua_s\delta\f.
	\end{align}
	On the one hand, by $\tau_s\in C\left(y_s,u_s\right)^+$ and the identification \cref{eq:tauident} we have for all $\delta\f\in\Y$ satisfying $0\leq\delta\f\leq\left\|\mathring\f\right\|_{L^\infty\left(\Omega\times\R^3\right)}$ a.e.,
	\begin{align*}
	\suma\taua_s\left(\delta\f-\f_s\right)\geq 0,
	\end{align*}
	which is \cref{lma:optcoJsii}. On the other hand, \cref{eq:taudef,eq:tauident} instantly yield \cref{eq:optcoJsiii} recalling the formula for $\mathcal J_s'$ from \cref{lma:Jsdiff}. 
	
	Setting $\delta u$ and all but one of the directions $\delta\f$, $\delta\f_+$, $\delta E$, and $\delta H$ to zero and the one remaining arbitrary, we conclude that the adjoint system \eqref{eq:Ads} holds. Note that a priori the $\psia_s$, $\vartheta_s^e$, and $\vartheta_s^h$ vanish for $t=T$ by definition of the test function space $\Lambda$.
	
	Finally, setting all directions but $\delta u$ to zero yields \eqref{eq:SCs}. Thus, also the proof of \cref{lma:optcoJsiii} is complete.
\end{proof}
\begin{remark}
	If for example $r=2$ and the boundary of $\Gamma$ is smooth, \eqref{eq:SCs} can easily be interpreted as the weak form of the second order PDE
	\begin{align*}
	\kappa_1\partial_t^2 u_s+\kappa_2\Delta_xu_s&=-4\pi\vartheta_s^e+\left(\nu_s^1\beta_1+\nu_s^2\beta_2+1\right)u_s&\mathrm{on}\ \left[0,T\right]\times\Gamma,\\
	\partial_tu_s\left(0\right)=\partial_tu_s\left(T\right)&=0&\mathrm{on}\ \Gamma,\\
	\partial_{n_\Gamma}u_s&=0&\mathrm{on}\ \left[0,T\right]\times\partial\Gamma.
	\end{align*}
	Here, $\partial_{n_\Gamma}$ denotes the directional derivative in the direction of the outer unit normal $n_\Gamma$ of $\partial\Gamma$.
\end{remark}
\subsection{Passing to the limit}
There remains to pass to the limit $s\to\infty$. A natural approach is to try to pass to the limit in the optimality conditions of \cref{prob:Ps}. This would require boundedness of the adjoint state in a certain norm. To this end, typically one needs to exploit some compactness result for the linearized PDE (system). In many situations, such results are available and one can then verify that the optimality conditions also hold in the limit, i.e., for a minimizer of the original problem. We refer to \cite{Lio85} for an abundance of examples of such PDEs.

However, for the Vlasov-Maxwell system no such results are available. In the author's opinion, the most problematic terms are the source terms on the right-hand side of \cref{eq:Ads4,eq:Ads5} which include $\partial_v\psia_s$, i.e., a derivative of the adjoint state. This is a structural problem arising because of the Vlasov-Maxwell system. Conversely, there are artificial problems, that is, the appearance of $\nu_s^1$, $\nu_s^2$, $\mathcal M_\alpha\f_s$, and $\taua_s$. They only appear because it is unknown whether the artificial constraints \cref{eq:constrfinfty,eq:constrener} in \cref{prob:P} (or then \cref{eq:constrfinfty,eq:constrener,eq:constrNa} in \cref{prob:Ps}) are automatically satisfied for some weak solution of \eqref{eq:WholeSystem} (or for a minimizing sequence of \cref{prob:Ps}). Especially $\taua_s$ is very irregular and there are no weak compactness results for the space which $\taua_s$ lies in.

Thus, we are not able to prove that a minimizer of \cref{prob:P} satisfies the desired optimality conditions, i.e., \eqref{eq:Ads} and \eqref{eq:SCs} with $s$ removed. Nevertheless, there holds the following, where we abbreviate $\min\mathcal J\coloneqq\mathcal J\left(y,u\right)$, $\left(y,u\right)$ being some minimizer of \cref{prob:P}:
\begin{theorem}\label{thm:passlimit}
	For each $s>0$, let $\left(y_s,u_s\right)\in\mathcal Y\times\mathcal U$ be a minimizer of \cref{prob:Ps}. Then
	\begin{align}\label{eq:Gdecay}
	\left\|\mathcal G\left(y_s,u_s\right)\right\|_{\Lambda^*}\leq\sqrt{2\min\mathcal J}s^{-\frac{1}{2}}
	\end{align}
	and there is a minimizer $\left(y_*,u_*\right)\in\mathcal Y\times\mathcal U$ of the original problem \cref{prob:P} such that, after choosing a suitable sequence $s_k\to\infty$, $\f_{s_k}\overset{(*)}\wto\f_*$ in $L^z\left(\left[0,T\right]\times\Omega\times\R^3\right)$ for $1<z\leq\infty$, $\f_{s_k,+}\to\f_{*,+}$ in $L^q\left(\gamma_{T}^+,d\gamm\right)$, $\left(E_{s_k},H_{s_k}\right)\wto\left(E_*,H_*\right)$ in $L^2\left(\left[0,T\right]\times\R^3;\R^6\right)$, and $u_{s_k}\to u_*$ in $\mathcal U$ for $k\to\infty$. Furthermore,
	\begin{align*}
	\lim_{k\to\infty}s_k\left\|\mathcal G\left(y_{s_k},u_{s_k}\right)\right\|_{\Lambda^*}^2=0.
	\end{align*}
\end{theorem}
\begin{proof}
	Let $\left(y,u\right)$ be some minimizer of \cref{prob:P}. Since this $\left(y,u\right)$ is also feasible for \cref{prob:Ps}, $\mathcal G\left(y,u\right)=0$, and since $\left(y_s,u_s\right)$ is a minimizer of \cref{prob:Ps}, there holds
	\begin{align}\label{eq:Gdecayest}
	\frac{s}{2}\left\|\mathcal G\left(y_s,u_s\right)\right\|_{\Lambda^*}^2\leq\mathcal J_s\left(y_s,u_s\right)\leq\mathcal J_s\left(y,u\right)=\mathcal J\left(y,u\right)=\min\mathcal J,
	\end{align}
	which implies \cref{eq:Gdecay} and that $\left(u_s\right)$ is bounded in $\mathcal U$ and $\left(f_{s,+}\right)$ in $L^q\left(\gamma_{T}^+,d\gamm\right)$. Thus, by \cref{eq:constrfinfty,eq:constrener} each $\left(\f_s\right)$ is bounded in any $L^z\left(\left[0,T\right]\times\Omega\times\R^3\right)$, $1\leq z\leq\infty$, and $\left(E_s,H_s\right)$ in $L^2\left(\left[0,T\right]\times\R^3;\R^6\right)$. Therefore, the asserted convergences hold true, at least weakly, if the sequence $\left(s_k\right)$ is suitably chosen. Since \cref{eq:constrener,eq:constrNa} are satisfied along the sequence, we can apply \cref{lma:Gwlsc} to obtain
	\begin{align*}
	\left\|\mathcal G\left(y_*,u_*\right)\right\|_{\Lambda^*}\leq\liminf_{k\to\infty}\left\|\mathcal G\left(y_{s_k},u_{s_k}\right)\right\|_{\Lambda^*}=0
	\end{align*}
	because of \cref{eq:Gdecay}. Hence, $\left(y_*,u_*\right)$ is feasible for \cref{prob:P}. By weak lower semi-continuity of any norm, there holds
	\begin{align}\label{eq:LimitMinimizer}
	\mathcal J\left(y_*,u_*\right)\leq\liminf_{k\to\infty}\mathcal J\left(y_{s_k},u_{s_k}\right)\leq\liminf_{k\to\infty}\mathcal J_{s_k}\left(y_{s_k},u_{s_k}\right)\leq\limsup_{k\to\infty}\mathcal J_{s_k}\left(y_{s_k},u_{s_k}\right)\leq\min\mathcal J,
	\end{align}
	where the last inequality is implied by \cref{eq:Gdecayest}. Consequently, $\left(y_*,u_*\right)$ is indeed a minimizer of \cref{prob:P} and equality holds in \cref{eq:LimitMinimizer}. Thus,
	\begin{align*}
	\mathcal J\left(y_*,u_*\right)=\liminf_{k\to\infty}\mathcal J\left(y_{s_k},u_{s_k}\right)\leq\limsup_{k\to\infty}\mathcal J\left(y_{s_k},u_{s_k}\right)\leq\limsup_{k\to\infty}\mathcal J_{s_k}\left(y_{s_k},u_{s_k}\right)=\min\mathcal J
	\end{align*}
	and also equality holds everywhere. This yields
	\begin{align}\label{eq:normlim}
	&\frac{1}{q}\suma \w\left\|\f_{*,+}\right\|_{L^q\left(\gamma_{T}^+,d\gamm\right)}^q+\frac{1}{r}\left\|u_*\right\|_{\mathcal U}^r=\mathcal J\left(y_*,u_*\right)=\lim_{k\to\infty}\mathcal J\left(y_{s_k},u_{s_k}\right)\nonumber\\
	&=\lim_{k\to\infty}\frac{1}{q}\suma \w\left\|\f_{s_k,+}\right\|_{L^q\left(\gamma_{T}^+,d\gamm\right)}^q+\frac{1}{r}\left\|u_{s_k}\right\|_{\mathcal U}^r.
	\end{align}
	Combining \cref{eq:LimitMinimizer,eq:normlim} implies
	\begin{align*}
	\lim_{k\to\infty}\frac{s_k}{2}\left\|\mathcal G\left(y_{s_k},u_{s_k}\right)\right\|_{\Lambda^*}^2=\lim_{k\to\infty}\mathcal J_{s_k}\left(y_{s_k},u_{s_k}\right)-\mathcal J\left(y_{s_k},u_{s_k}\right)=\mathcal J\left(y_*,u_*\right)-\mathcal J\left(y_*,u_*\right)=0.
	\end{align*}
	There remains to show that the convergences of $\f_{s_k,+}$ and $u_{s_k}$ are even strong. To this end, suppose that
	\begin{align*}
	\left\|f_{*,+}^{\alpha_0}\right\|_{L^q\left(\gamma_{T}^+,d\gamm\right)}<\limsup_{k\to\infty}\left\|f_{s_k,+}^{\alpha_0}\right\|_{L^q\left(\gamma_{T}^+,d\gamm\right)}
	\end{align*}
	for some $\alpha_0$. By weak lower semi-continuity of the remaining norms and by \cref{eq:normlim}, this implies
	\begin{align*}
	&\frac{1}{q}\suma \w\left\|\f_{*,+}\right\|_{L^q\left(\gamma_{T}^+,d\gamm\right)}^q+\frac{1}{r}\left\|u_*\right\|_{\mathcal U}^r\\
	&<\limsup_{k\to\infty}\left\|f_{s_k,+}^{\alpha_0}\right\|_{L^q\left(\gamma_{T}^+,d\gamm\right)}^q+\frac{1}{q}\sum_{\substack{\alpha=1\\\alpha\neq\alpha_0}}^N\w\liminf_{k\to\infty}\left\|\f_{s_k,+}\right\|_{L^q\left(\gamma_{T}^+,d\gamm\right)}^q+\liminf_{k\to\infty}\frac{1}{r}\left\|u_{s_k}\right\|_{\mathcal U}^r\\
	&\leq\limsup_{k\to\infty}\frac{1}{q}\suma \w\left\|\f_{s_k,+}\right\|_{L^q\left(\gamma_{T}^+,d\gamm\right)}^q+\frac{1}{r}\left\|u_{s_k}\right\|_{\mathcal U}^r=\frac{1}{q}\suma \w\left\|\f_{*,+}\right\|_{L^q\left(\gamma_{T}^+,d\gamm\right)}^q+\frac{1}{r}\left\|u_*\right\|_{\mathcal U}^r,
	\end{align*}
	which is a contradiction. Thus,
	\begin{align*}
	\left\|\f_{*,+}\right\|_{L^q\left(\gamma_{T}^+,d\gamm\right)}\leq\liminf_{k\to\infty}\left\|\f_{s_k,+}\right\|_{L^q\left(\gamma_{T}^+,d\gamm\right)}\leq\limsup_{k\to\infty}\left\|\f_{s_k,+}\right\|_{L^q\left(\gamma_{T}^+,d\gamm\right)}\leq\left\|\f_{*,+}\right\|_{L^q\left(\gamma_{T}^+,d\gamm\right)},
	\end{align*}
	whence
	\begin{align*}
	\left\|\f_{*,+}\right\|_{L^q\left(\gamma_{T}^+,d\gamm\right)}=\lim_{k\to\infty}\left\|\f_{s_k,+}\right\|_{L^q\left(\gamma_{T}^+,d\gamm\right)}
	\end{align*}
	for each $\alpha$. Similarly, there holds
	\begin{align*}
	\left\|u_*\right\|_{\mathcal U}=\lim_{k\to\infty}\left\|u_{s_k}\right\|_{\mathcal U}
	\end{align*}
	as well. Since we already have $\f_{s_k,+}\wto\f_{*,+}$ in $L^q\left(\gamma_{T}^+,d\gamm\right)$ and $u_{s_k}\wto u_*$ in $\mathcal U$, and since $L^q\left(\gamma_{T}^+,d\gamm\right)$ and $\mathcal U$ are uniformly convex, $\f_{s_k,+}$ even converges strongly to $\f_{*,+}$ and $u_{s_k}$ strongly to $u_*$.
\end{proof}
Note that the convergences of $\f_{s_k,+}$ and $u_{s_k}$ are strong, which is due to the fact that the original objective function $\mathcal J$ is an expression in $\f_+$ and $u$. Since the actual goal is to adjust $u$ suitably and $u$ is the only function which can be really adjusted from outside, it is no big drawback to have to consider \cref{prob:Ps} instead of \cref{prob:P}: As we have seen in \cref{thm:passlimit}, on the one hand $\left\|\mathcal G\left(y_s,u_s\right)\right\|_{\Lambda^*}$ decays with a certain rate to zero for $s\to\infty$, whence \eqref{eq:WholeSystem} is \enquote{almost} satisfied for a minimizer $\left(y_s,u_s\right)$ and $s$ large; on the other hand, first order optimality conditions for \cref{prob:Ps} have been established in \cref{thm:OCs} and optimal points of \cref{prob:Ps} converge (at least weakly) to an optimal point of \cref{prob:P} (along a suitable sequence), and the convergence of the controls is even strong. We can not expect to get convergence for the full limit $s\to\infty$ since minimizers of \cref{prob:P,prob:Ps} are not known to be unique due to the lack of convexity. Of course, the first order necessary optimality conditions are far from being sufficient.

\section{Some final remarks}\label{sec:finalremarks}
One can consider other optimal control problems than \cref{prob:P}, with a different objective function, for example a problem of tracking type:
\begin{align*}
\tilde{\mathcal J}\left(y,u\right)&=\mathcal J\left(y,u\right)+\suma\frac{b_\alpha}{2}\left\|\f-\f_d\right\|_{L^2\left(\left[0,T\right]\times\Omega\times\R^3\right)}^2+\frac{b_E}{2}\left\|E-E_d\right\|_{L^2\left(\left[0,T\right]\times\R^3;\R^3\right)}^2\\
&\phantom{=\;}+\frac{b_H}{2}\left\|H-H_d\right\|_{L^2\left(\left[0,T\right]\times\R^3;\R^3\right)}^2,
\end{align*}
where $b_\alpha,b_E,b_H>0$ are parameters and $\f_d,$ $E_d$, $H_d$ are desired states. Since this new objective function already grants coercivity in $\f$, $E$, and $H$ with respect to the $L^2$-norm, at first sight it seems that the artificial constraint \cref{eq:constrener} can be abolished. However, without this constraint, we can not pass to the limit in the term of \cref{eq:Maxwellweak1} with $j^\inte$ during an analogue proof of \cref{thm:ExistenceMinimizer}, since for this an $L_\kin^1$-estimate on $\f$ is necessary, cf. \cref{lma:jintconstr}. Thus, imposing \cref{eq:constrfinfty,eq:constrener} is still necessary. Analogues of \cref{thm:ExistenceMinimizer,thm:ExistenceMinimizerPs,thm:OCs,thm:passlimit} can be proved, and in \cref{thm:passlimit} the convergences of $\f_{s_k}$, $E_{s_k}$, and $H_{s_k}$ are also strong in $L^2$ because of the tracking terms in the new objective function.

We could also consider the case that we additionally try to control the system by inserting particles from outside, that is, adding some $\g\geq 0$ to the right-hand side of \cref{eq:WholeBoun} and treating them as controls as well. Then we add some norm of the $\g$ to the objective function as a penalization term. There occur two problems: First, since \cref{eq:constrfinfty} is still necessary and since we have to include $L^\infty$-norms of the $\g$ there on the right-hand side, the set of functions satisfying this new constraint is no longer convex. We can bypass this problem by imposing $L^\infty$-bounds on the $\g$ a priori, for example by imposing box constraints. Second, we have to add the $L_\kin^1$-norms of the $\g$ to the right-hand side of \cref{eq:constrener}. To be then able to pass to the limit in \cref{eq:constrener}, we need that the space the $\g$ lie in is compactly embedded in $L_\kin^1\left(\gamma_{T}^-,d\gamm\right)$ -- this is analogue to the consideration of $\mathcal U$ as the control space instead of simply $L^2$. That compact embedding is for example guaranteed by the restriction $\g\in H^1\left(\gamma_T^-\cap\left\{\left|v\right|\leq r\right\}\right)$ and $\g=0$ for $\left|v\right|>r$ with $r>0$ fixed. Another possibility is to impose an a priori bound on the $L_\kin^1$-norms of the $\g$, for example by imposing box constraints as above and a bound on the support of the $\g$ with respect to $v$, and then adding this a priori bound to the right-hand side of \cref{eq:constrener} instead of the $L_\kin^1$-norms of the $\g$.

In \cref{thm:passlimit}, a suitable sequence of optimal points of \cref{prob:Ps} converges to an optimal point of \cref{prob:P}, at least weakly, some components even strongly. However, we do not know if \textit{all} minimizers of \cref{prob:P} can be \enquote{obtained} in this way. In \cite{Lio85}, usually an approximate problem with an adaptive objective function is considered, in order to derive first order optimality conditions for \textit{any} given, fixed minimizer of \cref{prob:P}. Here, this means adding norms of $\f-\f_*$, $\f_+-\f_{*,+}$, $E-E_*$, $H-H_*$, and $u-u_*$ to $\mathcal J$. With an analogue of \cref{thm:passlimit}, one can then show that $\left(y_s,u_s\right)$ converges strongly to $\left(y_*,u_*\right)$ in a suitable norm, and this holds for the \textit{full} limit $s\to\infty$. However, this method is not constructive since one has to know $\left(y_*,u_*\right)$ a priori to consider the approximate problem, and thus in our case not reasonable; in general it is reasonable if one can pass to the limit in the first order optimality conditions.

\nocite{*}
\bibliography{ocrvma}
\bibliographystyle{plain}
\end{document}